\newif\ifdouble
\doubletrue
\ifdouble
\documentclass[journal,a4paper]{IEEEtran}
\else
\documentclass[12pt, draftclsnofoot, onecolumn]{IEEEtran}
\fi
\usepackage{subfigure}
\usepackage[utf8]{inputenc}
\usepackage[T1]{fontenc}
\usepackage{url}
\usepackage{ifthen}
\usepackage{cite}
\usepackage[cmex10]{amsmath} % Use the [cmex10] option to ensure complicance
\usepackage{comment}
\usepackage{amsthm}
\usepackage{amsmath,amssymb,amsfonts}
\usepackage{graphicx}
\usepackage{subfig}
\usepackage{algpseudocode}
\usepackage{algorithm}
\newtheorem{theorem}{Theorem}[section]
\newtheorem{corollary}{Corollary}[theorem]
\newtheorem{lemma}[theorem]{Lemma}

\newcommand{\paren}[1]{\left( #1 \right)}
\newcommand{\cbrace}[1]{\left\{#1\right\}}
\newcommand{\sbrace}[1]{\left[#1\right]}
\newcommand{\abs}[1]{\left| #1\right|}

\usepackage{mathtools}

\newcommand{\off}[1]{}

\newcommand\blfootnote[1]{%
  \begingroup
  \renewcommand\thefootnote{}\footnote{#1}%
  \addtocounter{footnote}{-1}%
  \endgroup
}

\newcommand{\eq}{\leftarrow}
\newcommand{\parent}[1]{\pi\paren{#1}}
\newcommand{\pj}{\parent{j}}

\newcommand{\cC}{{\cal C}}

\newcommand{\cE}{{\cal E}}

\newcommand{\cG}{{\cal G}}

\newcommand{\cN}{{\cal N}}

\newcommand{\cV}{{\cal V}}

\newcommand{\vc}{\vec{c}}

\newcommand{\vu}{\vec{u}}

\newcommand{\vx}{\vec{x}}
\newcommand{\vy}{\vec{y}}
\newcommand{\vz}{\vec{z}}

\newcommand{\vX}{\vec{X}}
\newcommand{\vY}{\vec{Y}}
\newcommand{\vZ}{\vec{Z}}

\newcommand{\vchat}{\hat{\vc}}
\newcommand{\vxhat}{\hat{\vx}}
\newcommand{\vyhat}{\hat{\vy}}
\newcommand{\vzhat}{\hat{\vz}}

\newcommand{\vZhat}{\widehat{\vZ}}

\newcommand{\gf}{\mathbb{F}}

\newcommand{\gftwo}{\gf_{2}}

\newcommand{\ex}[1]{\mathbb{E}\paren{#1}}
\newcommand{\var}[1]{\text{var}\paren{#1}}

\newcommand{\mutualinf}[2]{I\paren{#1;#2}}

\usepackage{xcolor}

\newcommand{\krd}[1]{{\color{red} KRD \color{red}{#1}}}

\newcommand{\SNR}{\text{SNR}}

\newcommand{\BLER}{B}

\begin{document}
\title{Noise Recycling}

% %%% Single author, or several authors with same affiliation:
% \author{%
%   \IEEEauthorblockN{Stefan M.~Moser}
%   \IEEEauthorblockA{ETH Zürich\\
%                     ISI (D-ITET)\\
%                     CH-8092 Zürich, Switzerland\\
%                     Email: moser@isi.ee.ethz.ch}
% }

%%% Several authors with up to three affiliations:
%\author{%
%  \IEEEauthorblockN{Stefan M.~Moser}
%  \IEEEauthorblockA{ETH Zürich\\
%                    ISI (D-ITET), ETH Zentrum\\
%                    CH-8092 Zürich, Switzerland\\
%                    Email: moser@isi.ee.ethz.ch}
%  \and
% \IEEEauthorblockN{Albus Dumbledore and Harry Potter}
%  \IEEEauthorblockA{Hogwarts School of Witchcraft and Wizardry\\
%                    Hogwarts Castle\\
%                    1714 Hogsmeade, Scotland\\
%                    Email: \{dumbledore, potter\}@hogwarts.edu}
%}

%%% Many authors with many affiliations:
 \author{%
   \IEEEauthorblockN{Alejandro Cohen\IEEEauthorrefmark{1},
                     Amit Solomon\IEEEauthorrefmark{1},
                     Ken R. Duffy\IEEEauthorrefmark{2},
                     and Muriel M\'edard\IEEEauthorrefmark{1}}\\
   \IEEEauthorblockA{\IEEEauthorrefmark{1}%
                     \textit{RLE, MIT}
                    Cambridge, MA 02139, USA,
                    \{cohenale,amitsol,medard\}@mit.edu}

   \IEEEauthorblockA{\IEEEauthorrefmark{2}%
                     \textit{Hamilton Institute
                    {Maynooth University, Ireland}, ken.duffy@mu.ie}}\\
 }

\maketitle

%%%%%%
%% Abstract:
%% If your paper is eligible for the student paper award, please add
%% the comment "THIS PAPER IS ELIGIBLE FOR THE STUDENT PAPER
%% AWARD." as a first line in the abstract.
%% For the final version of the accepted paper, please do not forget
%% to remove this comment!
%%

\begin{abstract}

We introduce Noise Recycling, a method that substantially enhances
decoding performance of orthogonal channels subject to correlated
noise without the need for joint encoding or decoding. The method
can be used with any combination of codes, code-rates and decoding
techniques. In the approach, a continuous realization of noise is
estimated from a lead channel by subtracting its decoded output
from its received signal. The estimate is recycled to reduce the
Signal to Noise Ratio (SNR) of an orthogonal channel that is experiencing
correlated noise and so improve the accuracy of its decoding. In
this design, channels only aid each other only through the provision
of noise estimates post-decoding.

For a system with arbitrary noise correlation between orthogonal
channels experiencing potentially distinct conditions, we introduce
an algorithm that determines a static decoding order that maximizes
total effective SNR with Noise Recycling. We prove that this
solution results in higher effective SNR than independent decoding,
which in turn leads to a larger rate region. When the noise is
jointly Gaussian, we establish that Noise Recycling employing this
static successive order enables higher code rates. We derive upper
and lower bounds on the capacity of any sequential decoding of orthogonal
channels with correlated noise where the encoders are independent
and show that those bounds are almost tight. We numerically compare
the upper bound with the capacity of jointly Gaussian noise channel
with joint encoding and decoding, showing that they match.

Simulation results illustrate that Noise Recycling can be employed
with any combination of codes and decoders, and that it gives
significant Block Error Rate (BLER) benefits when applying the
static predetermined order used to enhance the rate region. We
further establish that an additional BLER improvement is possible
through Dynamic Noise Recycling, where the lead channel is not
pre-determined but is chosen on-the-fly based on which decoder
provides the most confident decoding. Noise Recycling thus offers
significant decoding performance improvements without the need for
specific codes and decoders, or additional coordination between the
sender and receiver.  \end{abstract}


\section{Introduction}
The use of orthogonal channels is commonplace in applications from
wired to wireless channels. Examples include the wide-spread use
of orthogonal frequency division multiplexing (OFDM)\blfootnote{Parts of this
work \cite{cohen202noiserecycling} were presented at the IEEE International Symposium on Information
Theory, ISIT 2020.} \cite{NP00,
Yang05,Cim85}, and of orthogonal schemes in multiple access, such frequency
division multiplexing access (FDMA), time-division multiple access
(TDMA), or orthogonal code-division multiple access (CDMA), see,
for instance \cite{Gal85}. Let us consider the case of adjacent channels in a fading channel environment, for instance in OFDM or TDM.  In OFDM (TDM), channels or (channel uses) separated by less than a
coherence band (coherence time) will experience correlated fading \cite{Edetal88, Yanetal01, Hoe90, CVC01,CSS06, ZX06}, or equivalently, correlated noise.
 While
in theory joint decoding across channels can make use of such
correlation to improve performance \cite{cover2012elements,gallager1968information}, in practice it
is a challenge to implement owing to the computation complexity of
the decoding, and lack of compatibility with existing decoding
schemes.  Indeed, joint decoding runs counter to the reason for
seeking orthogonality in the first place.

\begin{figure}
    \centering
    \ifdouble
    \includegraphics[trim= 3cm 0.8cm 2cm 0.8cm,clip,width=1 \columnwidth]{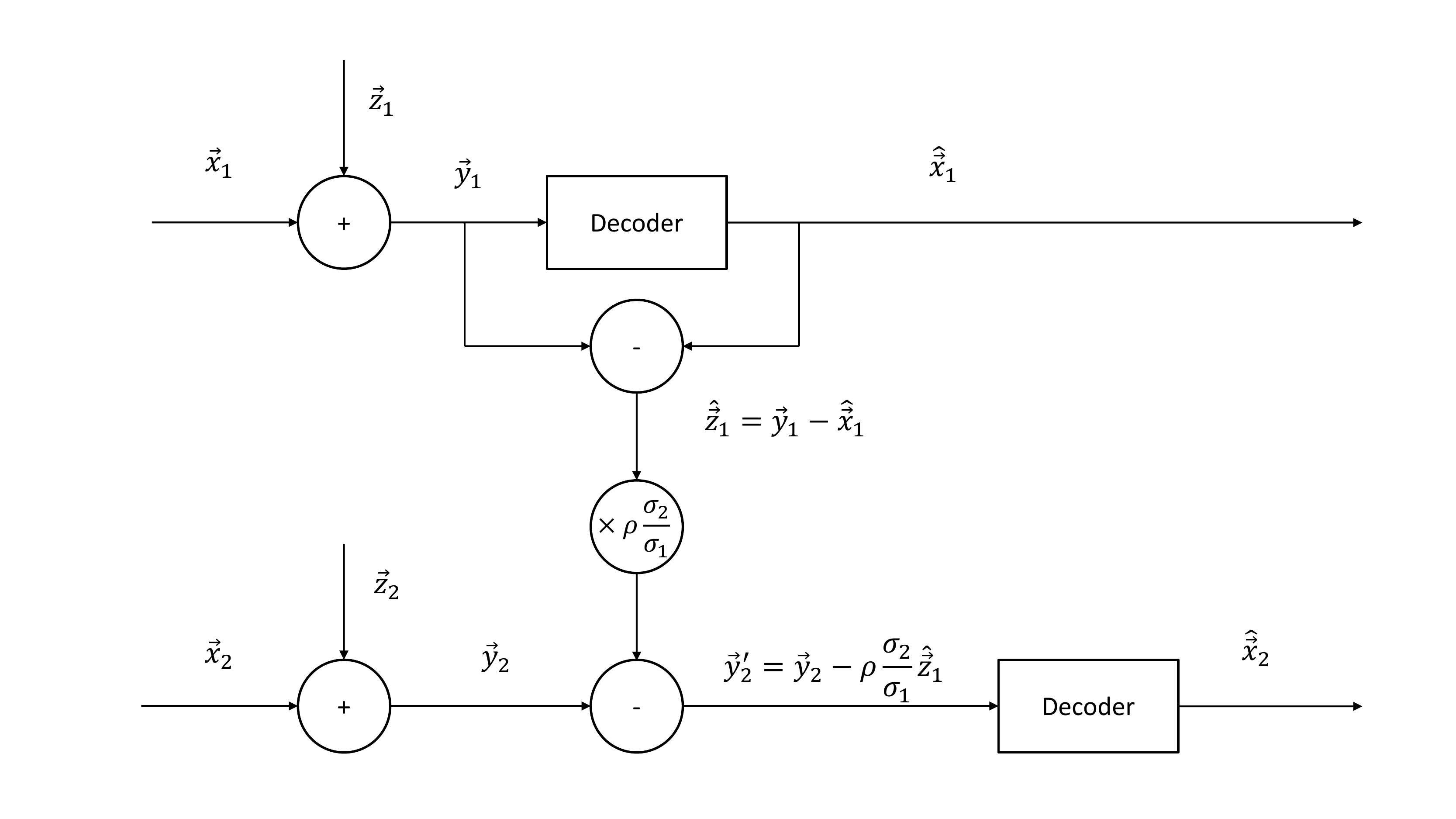}
    \else
     \includegraphics[trim= 3cm 0.8cm 2cm 0.8cm,clip,width=0.6 \columnwidth]{model}
    \fi
    \caption{In Noise Recycling, a noise estimate is created from a lead channel by subtracting its modulated decoding from the received signal. That estimate is used to reduce noise on a channel subject to correlated noise prior to decoding.}
    \label{fig:system_model}
    \vspace{-0.2in}
\end{figure}

Here we introduce a novel approach, Noise Recycling, that embraces
noise correlation to significantly improve decoding performance
while maintaining separate encoding and decoding across distinct
channels. The scheme is compatible with all encoding and decoding
schemes, and requires no coordination between channels. Its underlying
principle is that the realisation of noise experienced on that
channel is revealed when a decoding is correct, and that information
can be used to reduce the effective SNR for a neighbouring
communication.

Fig.~\ref{fig:system_model} provides an illustration of the
technique. Two independent channel inputs, $(\vx_1,\vx_2)$, from
potentially different codebooks are transmitted on orthogonal channels
and are corrupted by real-valued, mean-zero noise
$(\vZ_1,\vZ_2)$ with correlation $\rho$ and variances $\sigma_1^2$
and $\sigma_2^2$.  This results in correlated random real-valued
channel outputs $(\vY_1,\vY_2)= (\vx_1,\vx_2)+(\vZ_1,\vZ_2)$. For
a particular realization of outputs, $(\vy_1,\vy_2)$, a lead channel
is selected, say channel $1$, and $\vy_1$ is decoded to give
$\vxhat_1$. The decoder then estimates the noise realization
experienced on the lead channel by subtracting the decoded codeword
from the received signal $\vzhat_1 = \vy_1-\vxhat_1$. This noise
estimate is recycled to the second receiver who updates its channel
output prior to decoding via the Linear Least Square Estimator:
$\vy_2'=\vy_2-\rho\,\sigma_2/\sigma_1\vzhat_1$. If the lead
decoding is correct, which happens with probability one minus the
BLER, this eliminates part of the additive noise experienced on the
second channel $\vz_2$, before decoding, Noise Recycling results
in the revised second channel's output being a less noisy version
of the channel input $\vx_2$, which in turn leads to improved
decoding performance.

For code-rates below capacity, in the large code-length limit
concentration onto correct decodings occurs. This ensures that Noise
Recycling enables a larger rate region than independent decoding.
For BLER, there is a trade-off between the improved SNR that
comes with correct recycled noise and the disadvantage of passing
an erroneous estimate. A heuristic argument for why a gain is to
be expected is as follows. Consider two channels that have a block error
rate curves $B(\cdot)$  that is an increasing function of their noise variance
$\sigma^2$ when they are decoded independently, and assume that they experience
noise with correlation $\rho$. Suppose that the lead decoder either decodes
correctly or flags an error, as is the case in systems using a CRC
post-decoding for decoding validation. With probability $\BLER(\sigma^2)$
the lead channel would offer no noise for recycling and the BLER
performance of the second channel would also be $\BLER(\sigma^2)$.
If the lead channel decoded correctly, which occurs with
probability $1-\BLER(\sigma^2)$, the second channel would experience
a BLER of $\BLER(\sigma^2(1-\rho^2))$, giving, on average,
\begin{align}
\label{eq:bler}
\BLER(\sigma^2) \BLER(\sigma^2)  + \left(1-\BLER(\sigma^2)\right)
\BLER(\sigma^2(1-\rho^2)) <\BLER(\sigma^2).
\end{align}
Thus if the lead decoder never decodes erroneously, recycling is
always advantageous. Even in the absence of codeword validation
post-decoding, as errors are, in practice, rare, the benefits of
Noise Recycling are typically significant.

Noise Recycling is distinct from Interference Cancellation in
multiple access channels, where decoded codewords are subtracted
from received signals to remove interference \cite{cover2012elements,
Gal85, Ahl71, Lia72}. In Noise Recycling modulated decoded codewords are subtracted
from received signals to recover noise estimates which, owing to
correlation across channels, form a component of the noise in another
as-yet undecoded orthogonal channel. A proportion of the estimate
can, therefore, be subtracted from the received signal on the
orthogonal channel before decoding, reducing the latter's effective
noise. In non-orthogonal channels subject to both interference
and correlated noise, Noise Recycling and Interference Cancellation
could be used together.

In Section~\ref{sec:rate_gain} we mathematically establish the rate
gain that Noise Recycling provides over independently decoding the
channels. In Section~\ref{sub:decoding_order} we provide an algorithm
based on Maximum Directed Spanning Tree (MDST) that finds an optimal
static Noise Recycling decoding order for arbitrarily correlated
orthogonal channels in terms of maximizing the sum of effective
SNR. Assuming the noise is jointly Gaussian, in Section
\ref{subsec:achievable_region} we identify the extent of the enhanced
rate region that is made possible by the improved effective SNR.
We numerically evaluate rate gains over independent decoding of
channels, finding they improve both as correlation increases and
as the number of orthogonal channels increases for a given correlation.
For jointly Gaussian noise we then provide an upper bound on the
capacity of any pair of orthogonal channels with correlated noise
in Section~\ref{sec:UpperBound}. We compare the upper bound with
the achievability rate region given in
Section~\ref{subsec:achievable_region} for Noise Recycling, finding
that these bounds essentially coincide.  This upper bound is
numerically compared to the capacity of the channel using joint
decoding in which the encoders may cooperate \cite[Section
9.5]{cover2012elements} \cite{Tsy65, Tsy70} and, for the model considered in this work,
the upper bound and the capacity of the channel match.

In Section~\ref{sec:reliability_gain}, through simulation we determine
the BLER improvements yielded by Noise Recycling. We illustrate
that Noise Recycling can provide performance gains with any codes
at any rates using any decoders. We consider two distinct settings.
The one in Section \ref{sub:predetermined_order} is similar in
spirit to the approach presented in Section \ref{sec:rate_gain} and
employs the static channel recycling order that is determined based
on channel statistics via the MDST. The second setting, Section
\ref{sub:dynamic_order}, does not use a pre-determined decoding
order, but instead a dynamic per-realization one. The decoders of
orthogonal channels are first run in parallel.  The decoding that
results in the most confident decoding provides the first
estimate for Noise Recycling. While this approach is not designed
to provide rate gains, we show that it yields considerable BLER
improvements for both short and long codes.  Finally, we also
consider the possibility of re-recycling, where the lead channel
is itself fed a recycled noise estimate, finding that this
bootstrapping can improve the BLER performance of the leading
channel. A heuristic argument is provided in support of that
observation.

\section{System Model}\label{sec:model}
Let $x,\vx,X, \vX$ denote a scalar, vector, random variable, and
random vector, respectively. All vectors are row vectors. A linear
block code is characterized by $\sbrace{n,k}$, a code-length, $n$,
and a code-dimension, $k$, giving a rate $R=k/n$. The binary field
is denoted by $\gftwo$. Mutual information between $X,Y$ is denoted
by $\mutualinf{X}{Y}$.

We study an orthogonal channel system where $i\in\{1,\ldots,m\}$
messages, $\vu_i\in\gftwo^{k_i}$, are encoded into codewords
$\vc_i\in\gftwo^{n}$. Each codeword is modulated into a
block $\vx_i$ of real values and sent over continuous orthogonal channels subject
to additive real-valued noise. Channel outputs are
\begin{align*}
\vY_i=\vx_i+\vZ_i,
\end{align*}
where, $Z_{i}(l)$ the $l$-th element of $\vZ_i$, has mean
zero and variance
$\sigma_{i}^{2}$. For each $(i,j)$-th pair of orthogonal channels
the noise $(Z_i\paren{l},Z_j\paren{l})$ is assumed to follow a joint
distribution with correlation
\begin{align*}
\rho_{i,j} = \frac{\ex{Z_i(l)Z_j(l)}}{\sigma_i \sigma_j},
\text{ where }
\rho'_{i,j} = \rho_{i,j}\frac{\sigma_j}{\sigma_i}
\end{align*}
denotes the normalized correlation factor of the $j$-th channel
that is used in the linear least square estimator (LLSE) for
$Z_j\paren{l}$ given $Z_i\paren{l}$. The rate of the $i$-th code
is $R_i=k_i/n$, and the total rate is $R=\sum_{i=1}^m R_i$. Given
$(\vY_1,\ldots,\vY_m)$, the goal is to estimate $(\vc_1,\ldots,\vc_m)$
using $m$ distinct decoders.

When considering rate regions, we assume that each channel is decoded
only once. As a result, the SNR of the lead channel remains unchanged.
In the simulation results, however, we consider circumstances where
the decoders may operate repeatedly.

\section{Effective SNR Gain with Noise Recycling}\label{sec:rate_gain}
In this section we determine the rate region that can be achieved
by using Noise Recycling. We show that the total effective SNR
increases when Noise Recycling is applied. When Noise Recycling is
not used, each orthogonal channel has to be decoded independently
with a rate below that channel's capacity. In Noise Recycling, we
update the received signal of the $j$-th channel using noise estimated
from the $i$-th channel's decoding via the linear least squared
estimator
\begin{align*}
Y_j'=Y_j-\rho'_{i,j}Z_i=X_j+Z_j-\rho_{i,j}\frac{\sigma_j}{\sigma_i} Z_i.
\end{align*}
To see that it increases effective SNR with the recycled
noise is correct, note that
\begin{align}
\var{Z_j-\rho'_{i,j}Z_i} = \var{Z_j}(1-\rho^2_{i,j}),
\label{eq:varred}
\end{align}
and this is less than $\var{Z_j}$
so long as $\rho^2_{i,j}>0$. As a result, Noise Recycling
increases the orthogonal channel's effective SNR when the recycled
noise is correct. When the channel noise variances, $\var{Z_j}$,
are heterogeneous, eq. \eqref{eq:varred} demonstrates that Noise
Recycling impacts SNR in an asymmetric manner and so the order in
which noise is recycled impacts the total effective SNR.

\subsection{Optimal Decoding Order}\label{sub:decoding_order}

\begin{figure}
    \centering
    \includegraphics[trim=1.8cm 0.7cm 0cm 1cm,clip,scale=0.675]{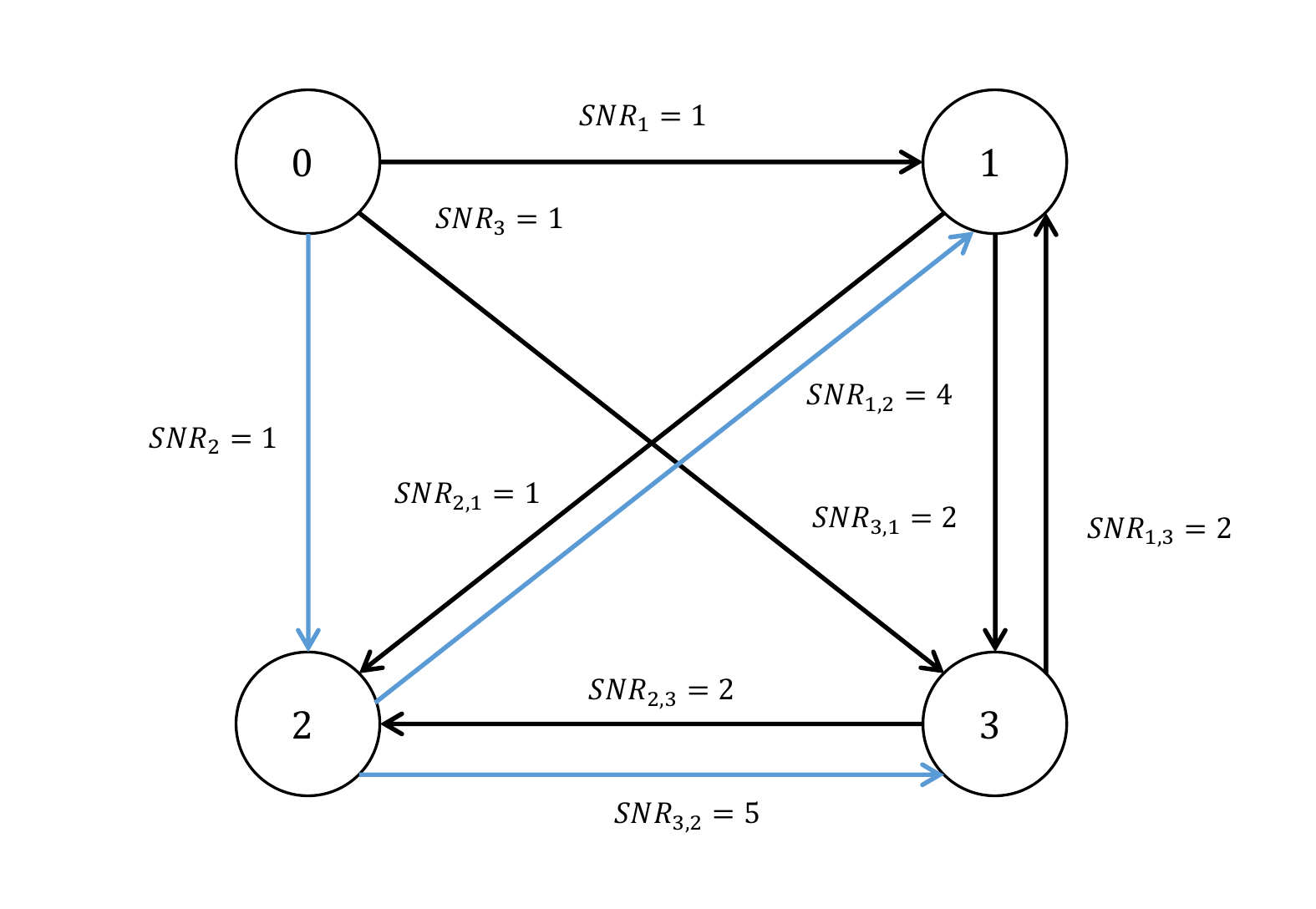}
    \caption{An example of the constructed graph of Section~\ref{sub:decoding_order}. There are $m=3$ orthogonal channels. The edges of a Maximum Directed Spanning Tree are painted in blue. In this example, the leading channel is the second orthogonal channel. The estimation of the second orthogonal channel is sent to the first and third orthogonal channels.}
    \label{fig:graph}
\end{figure}

We identify a solution to the problem of determining an optimal
static order for Noise Recycling using tools from graph theory. We
first construct an almost fully connected directed graph
$\cG=\paren{\cV,\cE}$ with $m+1$ nodes. Each node represents an
orthogonal channel, with one additional node called the zero node.
$\cG$ contains a directed edge from every node to every node, with
the exception of the zero node, which has only outgoing edges to
all other nodes. Each edge $\paren{i,j}$ is associated with a weight
\begin{equation*} w_{i,j}=\begin{cases} \SNR_{i,j} & i\neq 0,j\neq
0 \\ \SNR_{j} & i=0 \end{cases}, \end{equation*} where $\SNR_{i,j}$
is the effective SNR of orthogonal channel $i$ when using noise
recycled estimation of orthogonal channel $j$, and $\SNR_j$ is the
SNR of orthogonal channel $j$ without Noise Recycling.  Recall that
$\SNR_{i,j}$ is a function of the transmission power on channel
$i$, the statistics of the noise in channel $i$, $\sigma_i$, and
the correlation factor $\rho'_{i,j}$. The zero nodes helps us
represent the SNR associated with decoding without Noise Recycling.
An example of such graph with $m=3$ is shown in Fig.~\ref{fig:graph}.

\begin{figure}
    \centering
    \ifdouble
    \includegraphics[width=1 \columnwidth]{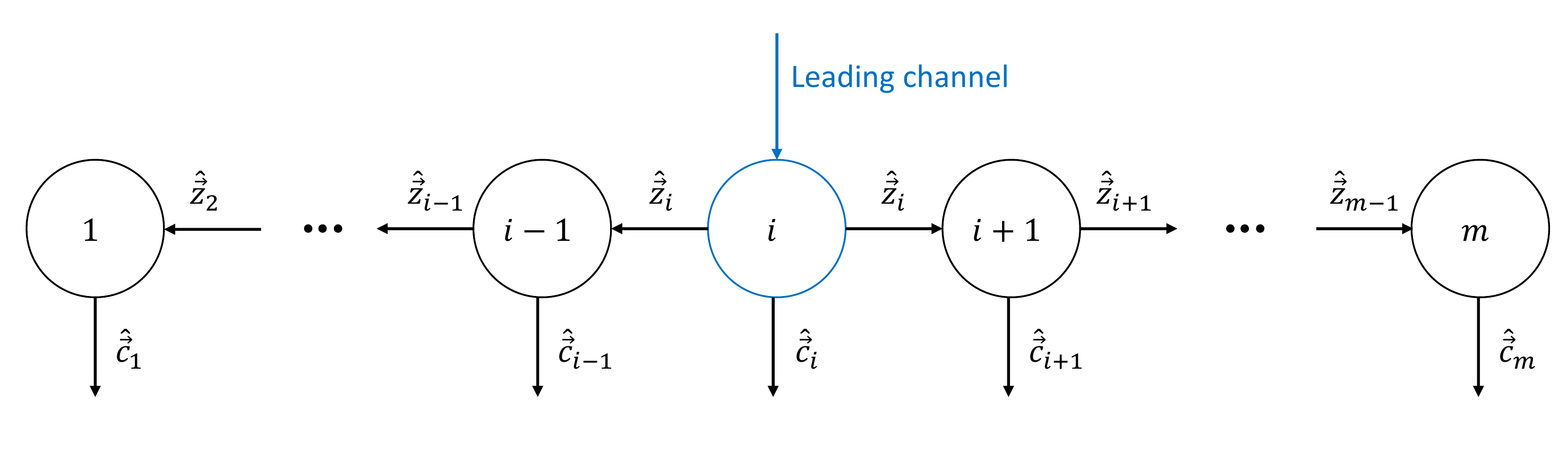}
    \else
    \includegraphics[width=0.6 \columnwidth]{gm_chain}
    \fi
    \caption{Gauss-Markov noise model, where correlation between
    adjacent orthogonal channels is $\rho$. A leading orthogonal
    channel $i$ is decoded first, and propagation of noise estimations
    follows to help decoding of all orthogonal channels.}
    \label{fig:gm_chain}
\end{figure}

\begin{figure*}
    \centering
    \ifdouble
    \includegraphics[width=2.05 \columnwidth]{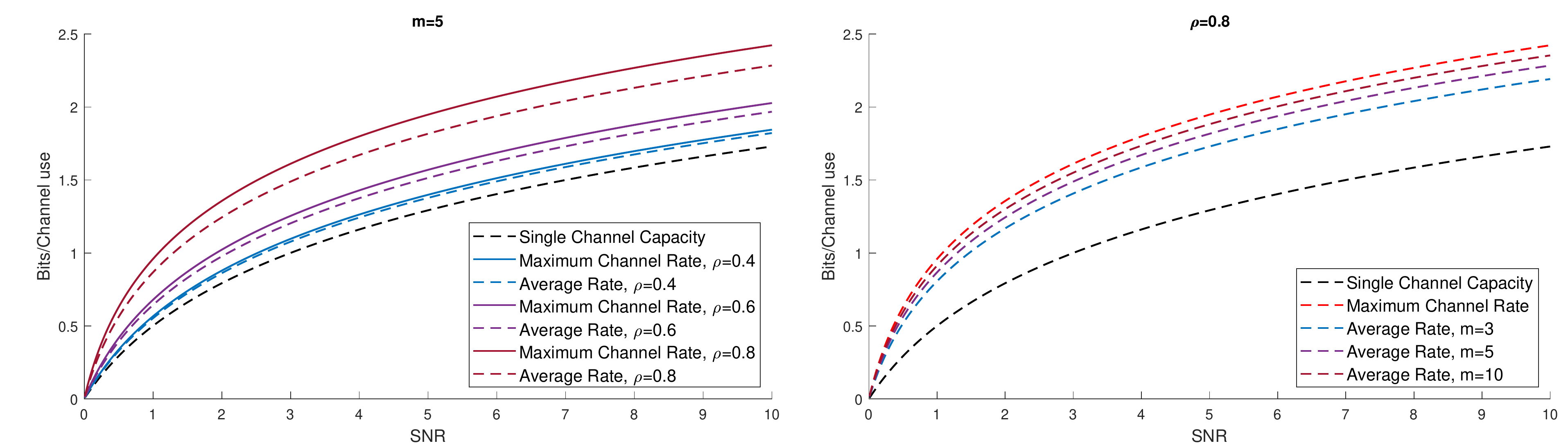}
    \else
    \includegraphics[width=1 \columnwidth]{capacity} \fi
    \caption{Achievable regime for the decoder of Section~\ref{sec:rate_gain}. Single Channel Capacity is  $C_1=C\paren{P/\sigma^2}=1/2\log\paren{1+P/\sigma^2}$, the capacity of an orthogonal channel that does not use recycled noise, as is the case for the first channel, $j=1$. Maximum Channel Rate is the rate of an orthogonal channel decoded with Noise Recycling $C_j=C\paren{P/(\paren{1-\rho^2}\sigma^2)}$, $j>1$. The average rate is the average rate per orthogonal channel, namely $(C_1+\paren{m-1}C_2)/m$.}
    \label{fig:capacity}
%    \vspace{-0.18in}
\end{figure*}

We show a decoding order which uses a MDST. Finding a MDST can be done efficiently using Edmond's algorithm~\cite{edmonds1967optimum}. Recall that a MDST $\cG'=\paren{\cV',\cE'}$, with a root node $r\in\cV$, has the following properties:
\begin{enumerate}
    \item $\cV'=\cV,\cE'\subseteq\cE$.
    \item $\forall i\in\cV',i\neq r: \exists$ a path from $r$ to $i$ in $\cE'$.
    \item The undirected version of $\cG'$ is a tree.
    \item The cost of $\cG'$, which is defined as the sum of all the weights of $\cE'$, is maximal among all possible Directed Spanning Trees of $\cG$.
\end{enumerate}
The decoding order we propose traverses a MDST of $\cG$ in a Breadth
First Search (BFS) fashion~\cite{cormen2009introduction} (although
other tree traversals are possible): First, all the nodes that are
connected to the zero node in $\cG'$ are decoded without noise
recycling. Then, every node that is connected to a previously decoded
node is decoded, using the noise recycled estimation of its parent
node. In the example of Fig.~\ref{fig:graph}, only node number 2
is connected to the zero node in the MDST. Hence, only orthogonal
channel number 2 is decoded without Noise Recycling. Then, orthogonal
channels number 1 and 3 are decoded with Noise Recycling, using the
estimation of the noise of the second orthogonal channel, as dictated
by the MDST. We see that the sum of the SNR with Noise Recycling
is $1+4+5$, which is the highest possible sum in this example, as
opposed to $1+1+1$ without Noise Recycling. Later in this section,
we show in Theorem~\ref{LB} that the achievable rates of this example
are $C\paren{1},C\paren{4},C\paren{5}$, unlike
$C\paren{1},C\paren{1},C\paren{1}$ without Noise Recycling, where
$C\paren{\cdot}$ is the capacity of the underlying channel as a function
of SNR. Pseudo-code
of the procedure is given in Algorithm~\ref{alg:general_noise}.
Note that lines~\ref{line:no_noise_recycling},\ref{line:noise_recycling}
use the channel output $\vy_j$, and $\vzhat_i,\rho'_{i,j}$ in
line~\ref{line:noise_recycling}. We prove that this solution is
optimal for finding a decoding order which maximizes the sum of
SNRs:

\begin{algorithm}
\caption{Noise Recycling Order}
\label{alg:general_noise}
\begin{flushleft}
        \textbf{Input:} $\vy_1,\ldots,\vy_m$\newline
        \textbf{Output:} $\vchat_1,\ldots,\vchat_m$
\end{flushleft}
\begin{algorithmic}[1]
\State $\cG\eq$ almost fully connected directed graph of Section~\ref{sub:decoding_order}
\State $\cG'\eq$ MDST of $\cG$
\State $\cN\eq\cbrace{r}$ \Comment{BFS queue}
\While{$\cN\neq\emptyset$}
\State $i\eq \text{pop}\paren{\cN}$ \Comment{Extract the head of $\cN$}
\For{all $j$ s.t. $\paren{i,j}\in\cE'$}
\If{$i=r$}
\State $\vchat_j\eq$ Decode $j$ without Noise Recycling\label{line:no_noise_recycling}
\Else
\State $\vchat_j\eq$ Decode $j$ with Noise Recycling of $i$\label{line:noise_recycling}
\EndIf
\State $\text{Add}\paren{\cN,j}$ \Comment{Add $j$ to the end of $\cN$}
\EndFor
\EndWhile
\end{algorithmic}
\end{algorithm}

\begin{lemma}\label{lem:mdst}
For arbitrary channels with arbitrary correlation $\rho_{i,j}$
between orthogonal channels $i,j$, Algorithm~\ref{alg:general_noise}
determines an order that maximizes the sum of effective SNR when
every orthogonal channel is decoded once.
\end{lemma}

\begin{proof}
The proof stems directly from the properties of a MDST. We have to prove the following:
\begin{enumerate}
    \item At least one channel is decoded without Noise Recycling.\label{prop:no_noise_recycling}
    \item Each orthogonal channel is decoded exactly once.\label{prop:decoded_once}
    \item The decoding order yields maximum total SNR.\label{prop:max_snr}
\end{enumerate}
Property~\ref{prop:no_noise_recycling} is satisfied as $r$ has at least one outgoing edge in $\cG'$, as a path from $r$ to $i$ exists for every $i\neq r$. Property~\ref{prop:decoded_once} is satisfied as $\cG'$ is a Directed Spanning Tree; Every node is reached, and no node can be reached more than once, as a tree has no cycles. Property~\ref{prop:max_snr} follows by the definition of a MDST.
\end{proof}

\subsection{Achievable Region}\label{subsec:achievable_region}
While Lemma~\ref{lem:mdst} holds in general, we present the
main theorem for jointly Gaussian noise. By improving the effective
SNR, a larger rate region is achievable than in the case where Noise
Recycling is not used.
\begin{theorem}\label{LB}
Assume a noise model with fixed $m$ orthogonal jointly Gaussian
channels, each with variance $\sigma_j^2$. For a given average power
constraint, $\mathbb{E}\paren{X_j^2}\leq P_j$, and any correlation
factor $\abs{\rho_j}<1$, the following region is achievable:
\begin{align*}
R_j<C\paren{\frac{P_j}{\paren{1-\rho_j^2}\sigma_j^2}}
\end{align*}
where
\begin{equation*}
\rho_j=\begin{cases}
\rho_{\pj,j} & \pj\neq r \\
0 & \pj=r
\end{cases},
\end{equation*}
$r$ is the root of the MDST (the zero node), $\pj$ is the parent
of $j$ in the MDST, $P_j/\sigma_j^2$ is the SNR and
$C\paren{P_j/\sigma_j^2}=1/2\log(1+P_j/\sigma_j^2)$ is the resulting
capacity.
\end{theorem}
\begin{proof}
The proof is given in Appendix~\ref{sec:code}.
\end{proof}

Note that in the case of joint
Gaussianity, the $j$-th noise vector $\vZ_j$ can be generated in
the following way, $Z_{j}(l)=\rho'_{i,j} Z_{i}(l)+\Xi_{i,j}(l)$,
where $Z_{i}(l)\sim\cN\paren{0,\sigma^2}$ and
$\Xi_{i,j}(l)\sim\cN\paren{0,\paren{1-\rho^2}\sigma^2}$ are independent
Gaussians.

 The Gauss-Markov (GM) process has been
used to model progressive decorrelation of fading  with growing
separation among channels \cite{NY07, AMM05, MG02, AG07, CZ04, Med00} in time,
frequency, or both, according to, say, Jakes's model \cite{Jak74}.
Mapping fading to an equivalent noise model leads naturally to a GM model of noise ~\cite{anderson2012optimal, Gob77, Mar94, Bro83, Mag65, Chr05, Kav99, Poo83}  to represent these effects.

\begin{corollary}

When the noise is GM, i.e. a symmetric
multivariate Gaussian with $\sigma_j=\sigma$, $\rho_{i,j}=\rho^{|i-j|}$, and equal average transmission power constraint $P_j=P$,
we get the following result: $\pj=j-1$ (where the root $r$ is the
zero node) and the achievable rates are
\begin{align*}
R_1<C\paren{\frac{P}{\sigma^2}},\quad R_j<C\paren{\frac{P}{\paren{1-\rho^2}\sigma^2}}, \text{ for all } j>1,
\end{align*}
where  $P/\sigma^2$ is the SNR and
$C\paren{P/\sigma^2}=1/2\log(1+P/\sigma^2)$ is the capacity.
\end{corollary}

The decoding order of a GM model is shown in Fig.~\ref{fig:gm_chain}
where the leading channel is $i=1$. The achievable rates are depicted
in Fig.~\ref{fig:capacity}, which illustrates the rate that can be
gained by Noise Recycling when compared to the case where channel
decoders operate independently. It is evident that there is a gap
between the single-channel capacity and the average rate that can
be achieved by decoders employing Noise Recycling. In particular,
there is a significant rate-gain even when the number of orthogonal
channels, $m$, or the noise correlation, $|\rho|$, is low.

We comment that Theorem~\ref{LB} can be naturally expanded to FDMA
and TDMA, where the rates have to be adjusted as in~\cite[Chapter
15.3.6]{cover2012elements}.

\begin{figure}
    \centering
    \ifdouble
    \includegraphics[width=1 \columnwidth]{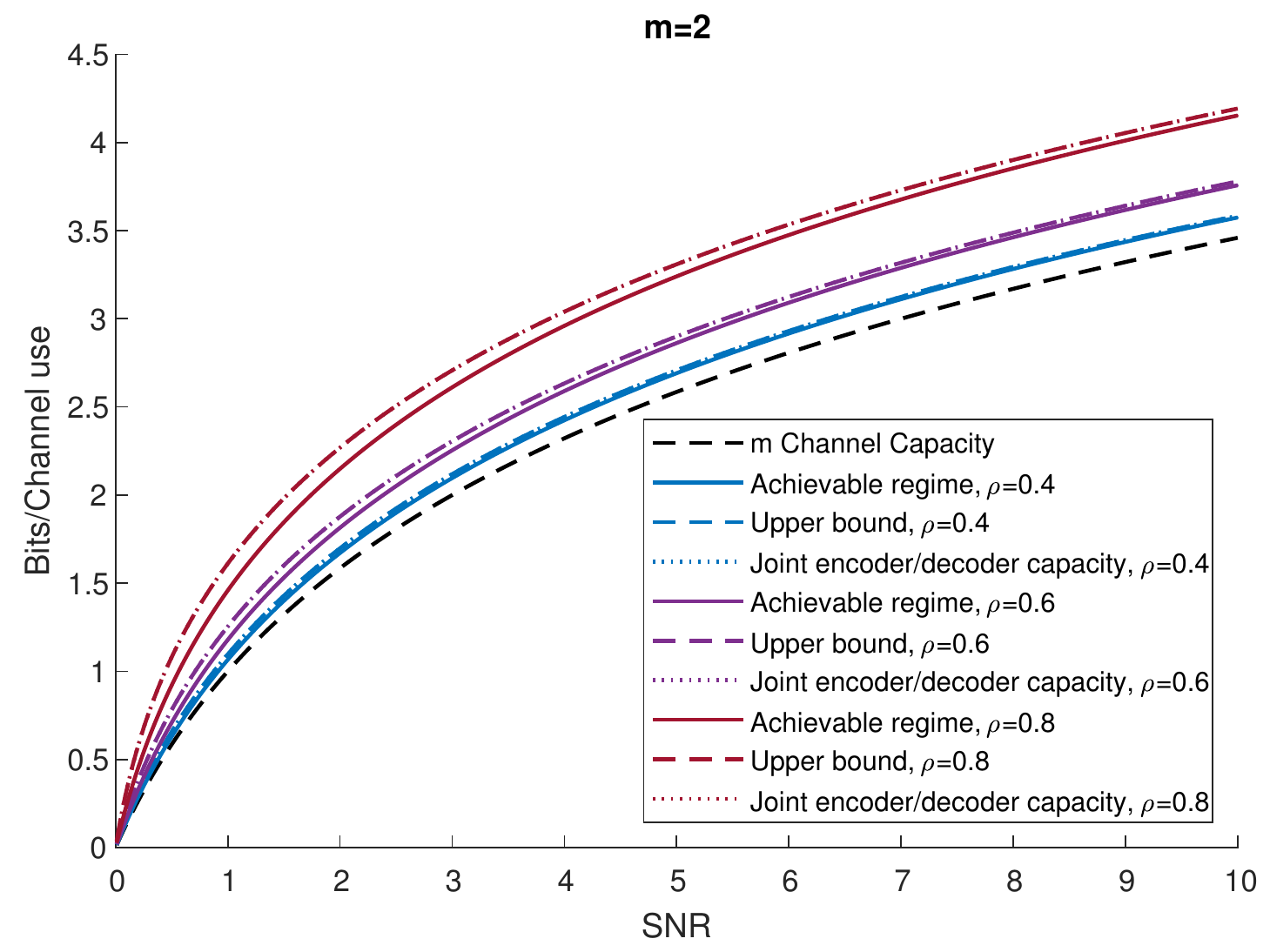}
    \else
    \includegraphics[width=0.6 \columnwidth]{UP}
    \fi
\caption{Upper bound vs. achievable regime for $m=2$ orthogonal
channel. The capacity of $m$ channels is $C_m= m\cdot C\paren{P/\sigma^2}$.
Achievable regime is the sum rate using Theorem~\ref{LB}. Upper
bound is the sum rate using the bound given in
Theorem~\ref{eq:upper_bound}. The cross channel joint encoder-decoder
channel capacity is $C_m =
\max_{\frac{1}{m}tr(\Lambda_{\underline{X}})\leq P}\frac{1}{2}\log
\left(\frac{\left| \Lambda_{\underline{X}} + \Lambda_{\underline{Z}}
\right|}{\left|\Lambda_{\underline{Z}} \right|} \right)$, where
$\Lambda_{\underline{X}}$ is the cross-correlation matrix of the
input energies and $\Lambda_{\underline{Z}}$ is the cross-correlation
matrix of the correlated noise, see \cite[Section 9.5]{cover2012elements}.}
    \vspace{1pt}
    \label{fig:UP}
\end{figure}

%%%%%%%%%%%%%%%%%%%%%%%%
\subsection{Upper Bound on Rate Region}\label{sec:UpperBound}

We now provide an upper bound on the rate region for orthogonal
channels with noise correlation $\rho_{i,j}$ in order to determine
what, if anything, is lost by our use of independent decoders with
Noise Recycling when compared with joint encoding and decoding.
Fig.~\ref{fig:UP} shows the upper bound rate region given in
Theorem~\ref{eq:upper_bound} (below). Moreover, we compare the upper
bound to the achievable regime using the Noise Recycling decoding
scheme provided in Theorem~\ref{LB}, the case where channel decoders
operate independently as given in \cite[Section 9.4]{cover2012elements},
and the case where the encoders may cooperate and there is a joint
decoder as given in \cite[Section 9.5]{cover2012elements}.

For a GM model, it is evident from the comparison presented in
Fig.~\ref{fig:UP} that there is a small gap that increases with
$\rho$ between the upper bound and the achievable regime using the
proposed Noise Recycling scheme. Yet, we note that for noise
correlation lower than 0.6 or for SNRs highers than 6, the bounds
almost match.  Furthermore, it is important to note that the upper
bound provided in this section for the case that the encoders are
independent, as considered in this paper, match the capacity of the
model where the encoders may cooperate and there is a joint decoder.
Recall that we consider a model where $Z_j\sim\cN\paren{0,\sigma^2}$
for each $j$-th orthogonal channel, as given in Section~\ref{sec:model}.

\begin{figure}
    \centering
    \ifdouble
    \includegraphics[width=1 \columnwidth]{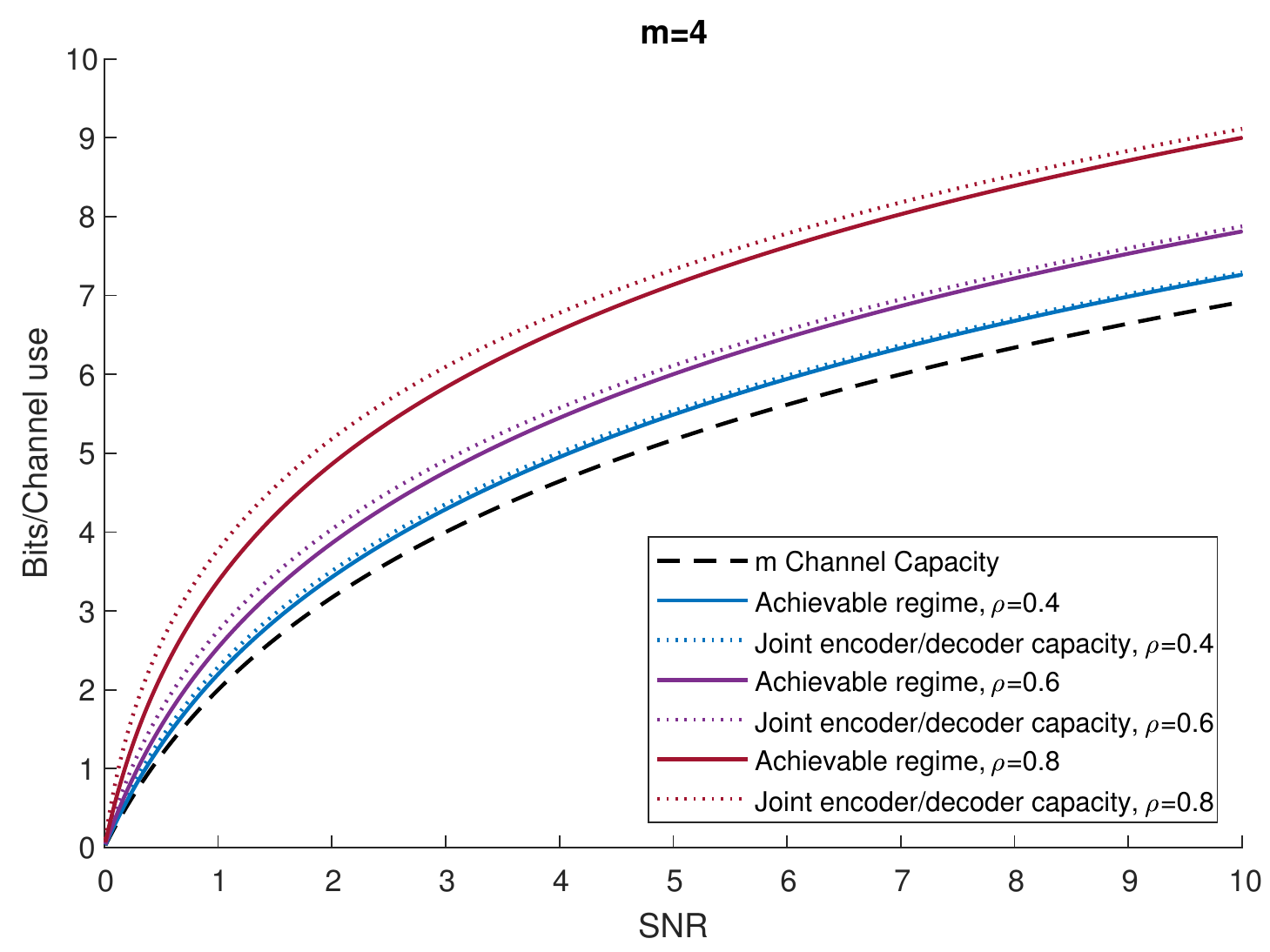}
    \else
    \includegraphics[width=0.6 \columnwidth]{capacity_m4_1}
    \fi
    \caption{Capacity and achievable region for $m=4$. We see that the achievable region is still close to the capacity, even when the number of channels increase.}
    \label{fig:UP2}
\end{figure}

\begin{theorem}\label{eq:upper_bound}
Assume a correlated noise model with fixed $m$ orthogonal channels,
each with variance $\sigma_{j}^{2}$, and correlation $\rho_{i,j}$
between each $(i,j)$ pair of orthogonal channels. For independent encoders
with a given average power constraint, $\mathbb{E}\paren{X_j^2}\leq
P_j$, and any correlation factor $\abs{\rho_{i,j}}<1$, the capacity
of any pair of orthogonal channels is upper bounded by,
\ifdouble
\begin{eqnarray}
    && \hspace{-0.7cm} C \leq \frac{1}{2}\log\left(1+\frac{P_i}{\sigma_{i}^{2}}\right) + \frac{1}{2}\log\left(\frac{P_j+\sigma_{j}^{2}}{P_j+\paren{1-\rho_{i,j}^{2}}\sigma_{j}^2}\right)\nonumber\\
    && \hspace{0.8cm} + \frac{1}{2}\log\left(1-\tilde{\rho}_{i,j}^2\right) + \frac{1}{2}\log \left(1+\frac{P_j}{\paren{1-\rho_{i,j}^2}\sigma_{j}^{2}}\right), \nonumber
\end{eqnarray}
\else
\[
     C \leq \frac{1}{2}\log\left(1+\frac{P_i}{\sigma_{i}^{2}}\right) + \frac{1}{2}\log\left(\frac{P_j+\sigma_{j}^{2}}{P_j+\paren{1-\rho_{i,j}^{2}}\sigma_{j}^2}\right)
     + \frac{1}{2}\log\left(1-\tilde{\rho}_{i,j}^2\right) + \frac{1}{2}\log \left(1+\frac{P_j}{\paren{1-\rho_{i,j}^2}\sigma_{j}^{2}}\right),
\]
\fi
where
\begin{align*}
\tilde{\rho}_{i,j}=\rho_{i,j}\frac{\sigma_i\sigma_j}{\sqrt{(P_i+\sigma_{i})(P_j+\sigma_{j})}}.
\end{align*}
\end{theorem}
{\em Proof:} The proof is given in Appendix~\ref{sec:UpperBound_proof}.

The upper bound provided in Theorem~\ref{eq:upper_bound} is for any
pair of orthogonal channels. Yet we note from Fig.~\ref{fig:UP2}
that the same trend concerning the tightness of the bounds is
maintained when the number of orthogonal channels increases, comparing
the achievable regime using Noise Recycling to the capacity with
encoders that may cooperate and a joint decoder is used.

%%%%%%%%%%%%%%%%%%%%%%%%
\section{Noise Recycling BLER Improvement}\label{sec:reliability_gain}

In Section~\ref{sec:rate_gain} we determined the rate-gains available
from Noise Recycling through the use of random codebooks and joint
typicality. Here we illustrate that BLER performance can be enhanced
by Noise Recycling when used with existing codes and decoders.
Simulations employ a jointly Gaussian channel model with Binary
Phase Shift Keying (BPSK) modulation.  We demonstrate the technique
with a diverse range of codes and decoders in terms of sizes and
rates to illustrate the general utility of the method.

\begin{figure}
    \centering
    \ifdouble
    \includegraphics[width=1 \columnwidth]{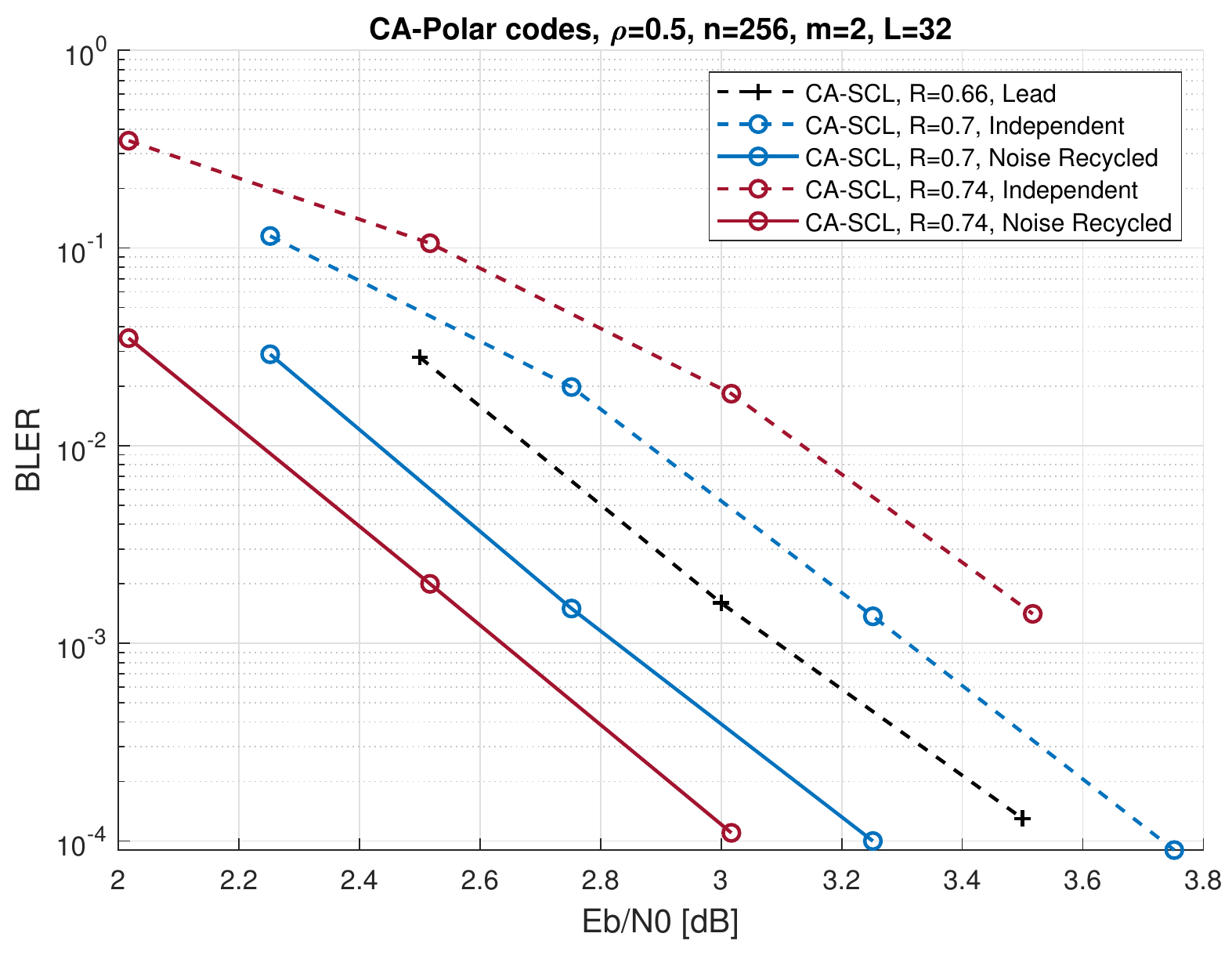}
    \else
    \includegraphics[width=0.6 \columnwidth]{bler_capolar_n_256_diff_rate_2}
    \fi
    \caption{BLER vs. Eb/N0 for $256$ bit CA-Polar codes decoded
    with CA-SCL, uses a list size of $L=32$, with and without Noise
    Recycling. Dashed lines correspond to independent decoding, and
    solid lines to decoding after Noise Recycling. The lead orthogonal
    channel is encoded with a rate $2/3$ code. The second channel
    uses either a rate $0.7$ or $0.74$ code.}
    \label{fig:bler_diff_rate_m_2}
    \vspace{-0.15in}
\end{figure}

For codes, we use:
\begin{itemize}
\item
CRC-Aided Polar (CA-Polar) Codes~\cite{tal2011list,tal2015list},
which are Polar codes~\cite{arikan2008channel,arikan2009channel}
with an outer CRC code. These have been proposed for all control
channel communications in 5G NR~\cite{3gppcapolar1}.
\item
Low Density Parity Check (LDPC) codes~\cite{gallager1962low}, which
have been proposed for all data channel communications in 5G
NR~\cite{3gppcapolar1} and typically have long code-lengths.
\item
Random Linear Codes (RLCs), which have long-since been known to be
capacity achieving~\cite{gallager1973random}, but have been little
investigated owing to the lack of availability of decoders until
the recent development of Guessing Random Additive Noise Decoding
(GRAND) \cite{duffy2018guessing,duffy2019capacity,duffy2019SRGRAND,duffy20195g}.
\end{itemize}
For decoders, we use the state-of-the-art CA-Polar-specific CRC-Aided
Successive Cancellation List decoder (CA-SCL)
~\cite{tal2011list,niu2012crc,tal2015list,balatsoukas2015llr,liang2016hardware}
and the LDPC decoder using Belief Propagation~\cite{gallager1962low},
both as implemented in MATLAB's Communications Toolbox. We also use
two soft-information GRAND variants \cite{solomon2020,duffy2020},
which are well suited to short, high-rate codes. While most decoders
are tied to specific code-book constructions, both of these can
decode any block code, including CA-Polar codes and RLCs.

\begin{figure}
    \centering
    \ifdouble
    \includegraphics[width=1\columnwidth]{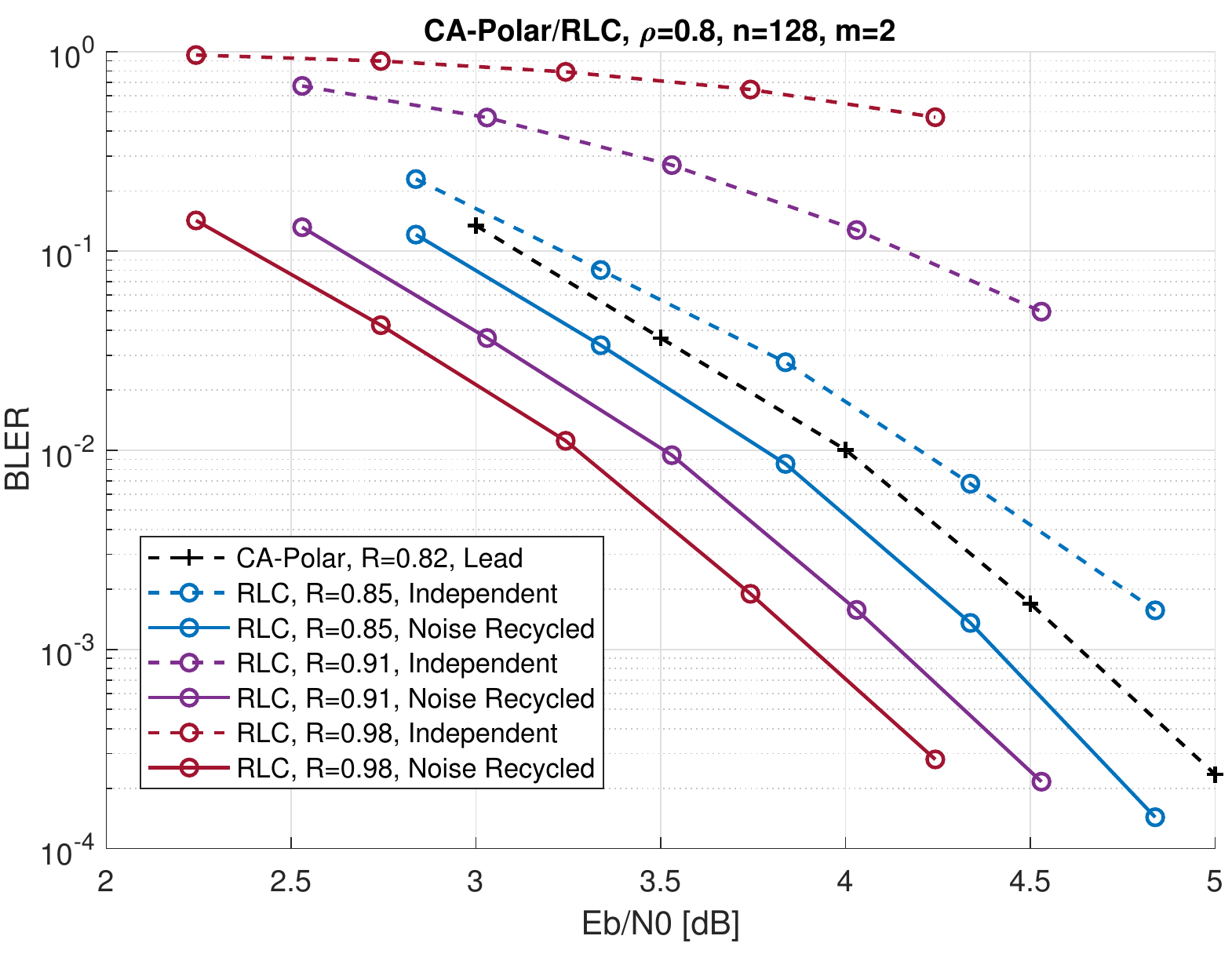}
    \else
    \includegraphics[width=0.6\columnwidth]{bler_caploar_rlc_n_128_v3}
    \fi
    \caption{BLER vs. Eb/N0 for codes of length $n=128$
    decoded with ORBGRAND with and without Noise Recycling. Dashed lines correspond to independent decoding,
    and solid lines to decoding after Noise Recycling. Data on the
    lead orthogonal channel is encoded with a rate $0.82$ CA-Polar code. The
    second channel uses rate $0.85$, $0.91$ or $0.98$ RLCs.}
    \label{fig:bler_diff_rate_m_2_capolar_rlc}
    \vspace{-0.2in}
\end{figure}

\subsection{Static Noise Recycling}
\label{sub:predetermined_order}

We first consider a sequential decoding scheme akin to the one
described in Section~\ref{sec:rate_gain} where a lead channel is
selected {\it a priori} and decoded. A subsequent channel that has a
higher rate is then decoded using noise recycled information.
Block-errors are counted separately on each channel.

In the first simulation, the lead channel encodes its data using
a CA-Polar code $\sbrace{256,170}$ with rate $R_1\approx 2/3$. The
second, orthogonal channel uses a higher rate CA-Polar code, either
$\sbrace{256,180}$ or $\sbrace{256,190}$ giving $R_2\approx 0.7$
or $0.74$ respectively. Noise variances are the same on both channels
and the noise correlation, which is set to $\rho=0.5$, is known to
the second decoder. Both channels are decoded with CA-SCL, with the
second channel benefiting from Noise Recycling.

\begin{figure*}
    \centering
    \ifdouble
    \includegraphics[width=2.05 \columnwidth]{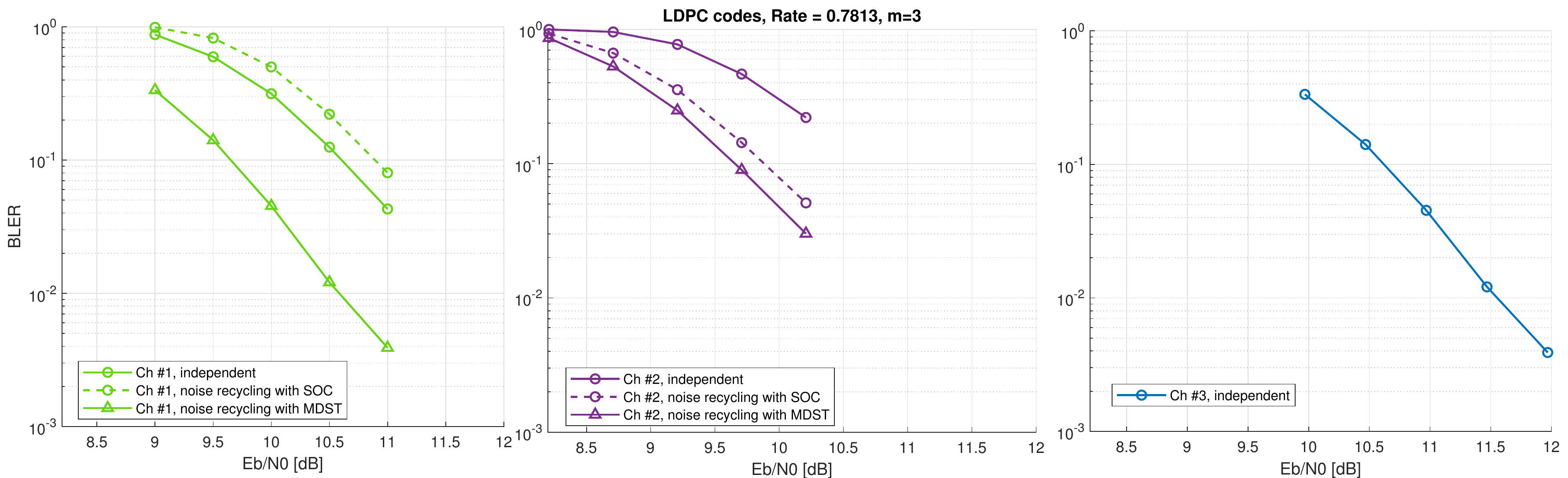}
    \else
    \includegraphics[width=1 \columnwidth]{LDPC_diff_var_energy5_static}
    \fi
    \caption{BLER vs Eb/N0 for $\sbrace{12800,10000}$ LDPC codes
    decoded with Belief Propagation and static Noise Recycling,
    for $m=3$ orthogonal channels. The  correlation values between
    the different channels are $\rho_{1,2}=0.6$, $\rho_{1,3}=0.8$,
    and $\rho_{2,3}=0.4$. The values of the variance of the orthogonal
    channels are $\sigma_1^2=\sigma_f^2$,
    $\sigma_2^2=1.2\,\sigma_f^2$ and $\sigma_3^2=0.8\,\sigma_f^2$.
    The third channel is chosen as the leading channel, while the order of
    subsequent channels are chosen statically using either Algorithm 1 (MDST), or a simple sequential orthogonal channel (SOC) order.}
    \label{fig:cov_unnorm}
    \vspace{-0.2in}
\end{figure*}

Fig.~\ref{fig:bler_diff_rate_m_2} reports BLER vs Eb/N0 for these
medium length codes. The black dashed line corresponds to the
lead channel, while the dashed blue and red lines give the performance
curves should Noise Recycling not be used, corresponding to independent
decoding of all channels.  As the second orthogonal channel runs
at a higher rate than the lead channel, if decoded independently
the second channel would experience higher BLER than the lead
channel. The solid blue and red lines report the performance of the
second decoder given noise recycling. Despite using a higher rate
code than the lead channel, with Noise Recycling the second channel
experiences better BLER vs Eb/N0 performance. Notably, owing to the
better Eb/N0 (i.e. the energy per {\bf information} bit used in the
transmission) that comes from running a higher rate code, the rate
$0.74$ code provides better BLER than the rate $0.7$ code. For a
commonly used target BLER of $10^{-2}$, Noise Recycling results in
$\approx 1$ dB gain for the $\sbrace{256,190}$ code.

Fig.~\ref{fig:bler_diff_rate_m_2_capolar_rlc} reports an analogous
simulation for short, high-rate codes, additionally illustrating
that the methodology is agnostic to having distinct codes on different
channels. With $\rho=0.8$, the lead channel's code is a $\sbrace{128,105}$
CA-Polar code of rate $R_1=0.82$ and the second channel is one of
three RLCs with a rate ranging from $0.85$ to $0.98$. Both channels
are decoded with the recently proposed soft detection decoder Ordered
Reliability Bits Guessing Random Additive Noise Decoding
(ORBGRAND)~\cite{duffy2020}.  As with all the GRAND algorithms, it
can decode any code, making it viable for use with RLCs, which
encompass all possible codes. A similar phenomenology to the previous
figure can be seen, where the impact of Noise Recycling is even
more dramatic, allowing the second channel code to use reliably a
much higher rate than the lead channel.

Our last simulation of this section illustrates that noise recycling
also provides BLER gains for long codes. In Section~\ref{sec:rate_gain}
we proved that MDST determined an optimal static order for decoding
in the presence of asymmetric channel conditions. Here we demonstrate
that decoding order plays a significant role in decoding performance
for heterogeneous channel conditions.

We simulated Noise Recycling for $m=3$ orthogonal channels subject
to asymmetric, correlated Gaussian noise, with correlations
$\rho_{1,2}=0.6$, $\rho_{1,3}=0.8$ and $\rho_{2,3}=0.4$. We define
a common noise factor, $\sigma_f^2$, so the noise variance of the
channels is given by $\sigma_1^2=\sigma_f^2$, $\sigma_2^2=1.2$
$\sigma_f^2$ and $\sigma_3^2=0.8$ $\sigma_f^2$. Each channel used
a long $\sbrace{12800,10000}$ 5G LDPC code, decoded with Belief
Propagation, and two static Noise Recycling orders are illustrated.
One is the optimal SNR order, as determined by
Algorithm~\ref{alg:general_noise} (MDST), that is advocated here,
which is to lead with Channel 3, then recycle noise before decoding
Channel 1, and then recycle again before decoding Channel 2.  The
other is an order that does not take the channel conditions into
account. It contrast, it again leads with Channel 3, but then
recycles noise to Channel 2, before decoding and recycling noise
to Channel 1 for its decoding.

Fig.~\ref{fig:cov_unnorm} reports BLER performance of these two
static orders. As Channel 3, which operates at the best SNR, leads
for both orders, it experiences fixed performance. Using
MDST then sees a gain of $\sim$1.25dB on Channel 1, followed by a gain
of $\sim$1dB on Channel 2. In contrast, in the non-optimal order,
while Channel 2 sees a $\sim$0.75dB gain, it is operating at the
lowest SNR and the noise that it passes for recycling to Channel
1 is deleterious, leading to $\sim$0.25dB loss. Overall, a > 2dB gain
is possible with MDST, while using an inappropriate order sees
a total gain of only 0.5dB.

\subsection{Dynamic Noise Recycling}\label{sub:dynamic_order}

While previous sections identified rate and BLER improvements that
are available from running a pre-determined lead channel with a
lower rate code so that an accurate inference of a noise realization
could be obtained to aid the signal at a higher rate second channel,
here we consider an alternate design that can lead to a significant
additional gain with both short and long codes.

\begin{figure*}
    \centering
    \ifdouble
    \includegraphics[width=2 \columnwidth]{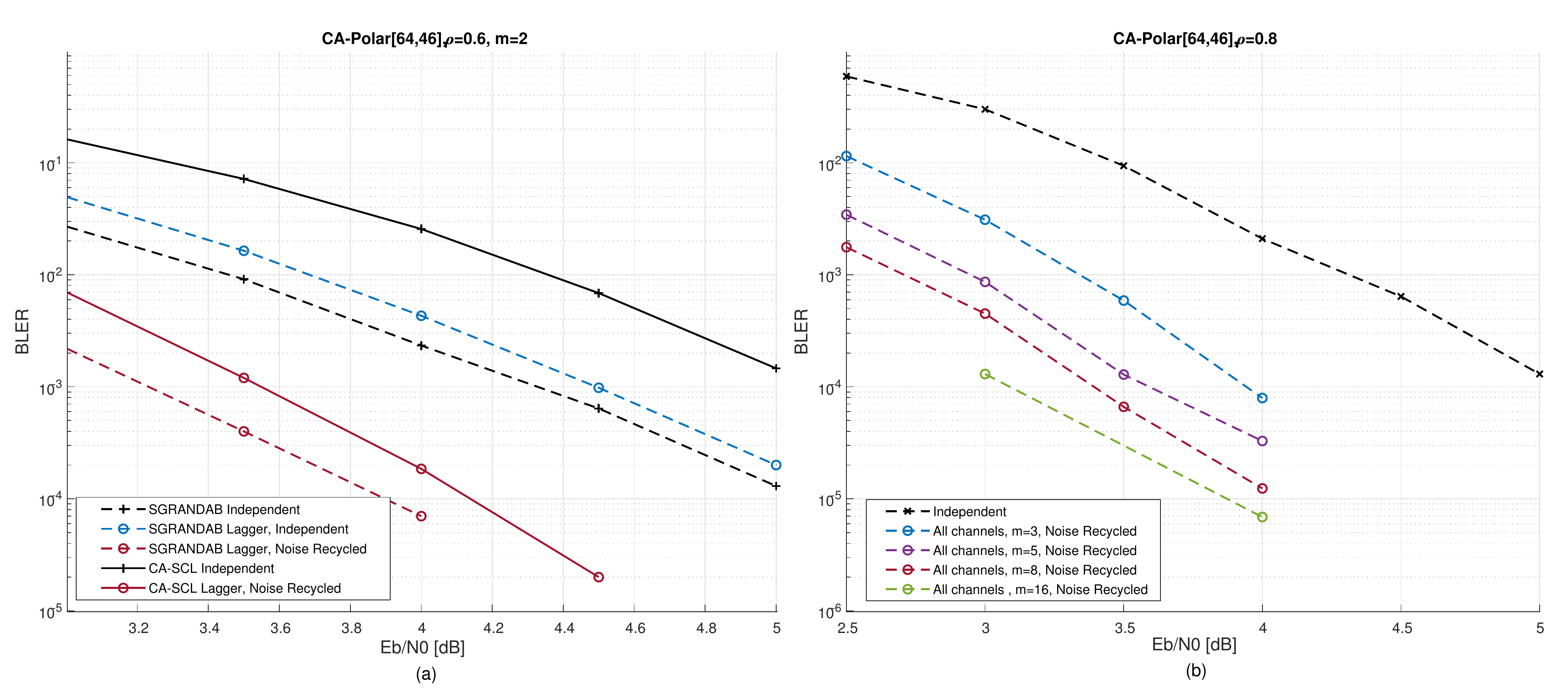}
    \else
    \includegraphics[width=1 \columnwidth]{bler_same_rate_MG_2}
    \fi
    \caption{BLER vs Eb/N0 for $\sbrace{64,46}$ CA-Polar codes with
    dynamic Noise Recycling using SGRANDAB with an abandonment
    threshold of $b=10^6$ for decoding in the first phase. In (a), either
    SGRANDAB or CA-SCL, which uses a list size of $L=32$, is used
    to decode the remaining channels after Noise Recycling, while in (b) all decoding
    is done using SGRANDAB.}
    \label{fig:bler_same_rate_m_mg}
\end{figure*}

The principle behind Dynamic Noise Recycling is that all orthogonal
channels initially attempt to decode their outputs contemporaneously.
In principle one would then wish to select the most confident
decoding to lead the Noise Recycling for that particular realization.
In practice, it may be necessary to use a post-decoding soft-information
proxy for that confidence. For example, regardless of the decoders
employed, one could select the least energetic estimated noise
sequence or the most likely noise sequence. While some decoders,
such as those based on the GRAND paradigm, themselves provide soft
information on the confidence of their decoding. The Dynamic Noise
Recycling decoding procedure with only one round of competition is
described in Algorithm~\ref{alg:same_rate}.

\begin{algorithm}
\caption{Dynamic Noise Recycling}
\label{alg:same_rate}
\begin{flushleft}
        \textbf{Input:} $\vy_1,\ldots,\vy_m$\newline
        \textbf{Output:} $\vchat_1,\ldots,\vchat_m$
\end{flushleft}
\begin{algorithmic}[1]
\State Decode orthogonal channel outputs
\State $i\gets\text{ \# of most confident decoding}$
\State $\vchat_i\gets \;i\text{-th decoded codeword}$
\State $\vxhat_i\gets \;\text{modulation of }\vchat_i$
\State $\vzhat_i\gets \vy_i-\vxhat_i$
\For{$j=1\rightarrow \max\cbrace{m-i,i-1}$}
\If{$i+j\leq m$}
\ifdouble
\State $\sbrace{\vchat_{i+j},\vzhat_{i+j}}\gets\text{DecodeAndEst.}$\newline$\paren{\vy_{i+j},\vzhat_{i+j-1},i+j-1,i+j}$
\else
\State $\sbrace{\vchat_{i+j},\vzhat_{i+j}}\gets\text{DecodeAndEst.}$$\paren{\vy_{i+j},\vzhat_{i+j-1},i+j-1,i+j}$
\fi
\EndIf
\If{$i-j\geq 1$}
\ifdouble
\State $\sbrace{\vchat_{i-j},\vzhat_{i-j}}\gets\text{DecodeAndEst.}$\newline$\paren{\vy_{i-j},\vzhat_{i-j+1},i-j+1,i-j}$
\else
\State $\sbrace{\vchat_{i-j},\vzhat_{i-j}}\gets\text{DecodeAndEst.}$$\paren{\vy_{i-j},\vzhat_{i-j+1},i-j+1,i-j}$
\fi
\EndIf
\EndFor
\State \Return $\vchat_1,\ldots,\vchat_m$
\end{algorithmic}
\begin{algorithmic}
\Procedure{DecodeAndEst.}{$\vy,\vzhat,i,j$}
\State $\vyhat\gets\vy-\rho'_{j,i}\vzhat$
\State Decode orthogonal channel $j$ using $\vyhat$
\State $\vchat\gets\; \text{decoded codeword}$
\State $\vxhat\gets \;\text{modulation of }\vchat$
\State $\vzhat\gets \vy-\vxhat$
\State \Return $\vchat,\vzhat$
\EndProcedure
\end{algorithmic}
\end{algorithm}

As an example, suppose there are 3 orthogonal channels. At the
first step, all decode in parallel. If decoder 2 provides the most
accurate decoding, it is selected as the lead, providing an estimate
$\vzhat_2$ to the decoders 1 and 3, which repeat the process. Note
that mixing-and-matching of decoders, even at different stages of
the dynamic Noise Recycling, is possible.

We first we consider decoders in which the speed of decoding provides
a measure of confidence in their decoding accuracy. Soft GRAND with
ABandonment (SGRANDAB)~\cite{solomon2020} has that feature. SGRANDAB
aims to identify the noise that corrupted a transmission from which
the codeword can be inferred, rather than identifying the codeword
directly and can decode any block code. It does this by removing
possible noise effects, from most likely to least likely as determined
by soft information, from a received signal and querying whether
what remains is in the codebook. The first instance that results
in success is a maximum likelihood decoding. If no codeword is found
before a given number of codebook queries, SGRANDAB abandons decoding
and reports an error. SGRANDAB reports the number of code-book
queries made until a code-book element was identified, and fewer
queries can being used as a proxy for a more confident decoding.

We simulated Dynamic Noise Recycling in GM channels using a
$\sbrace{64,46}$ CA-Polar code. We first consider the method on
$m=2$ orthogonal channels with $\rho=0.6$, where the initial decoding
is performed using SGRANDAB on both channels, and the noise-recycled
decoding is performed using either SGRANDAB or CA-SCL.
Fig.~\ref{fig:bler_same_rate_m_mg} (a) reports BLER performance.
As without Noise Recycling SGRANDAB results in lower BLER than
CA-SCL, it provides a correct noise sequence more often. The second
decoder can use SGRANDAB, as during the initial phase, or instead
use CA-SCL. Either way, the remaining channel benefits significantly
from Noise Recycling, with a gain of more than 1 dB, even for codes
of the same rate, by Dynamic Noise Recycling.
Fig.~\ref{fig:bler_same_rate_m_mg} (b) reports the BLER of an
SGRANDAB decoded channel without Noise Recycling.  For $\rho=0.8$,
and $m=3$, 5 or 8, Dynamic Noise Recycling is employed for one
round, and then Noise Recycling, where all decoders use SGRANDAB.
Again, this shows a significant improvement in BLER for all values
of $m$. For example there is a gain of about 1.7 dB for $m=8$ at a
target BLER of $10^{-4}$.

\begin{figure}
    \centering
    \ifdouble
    \includegraphics[width=1 \columnwidth]{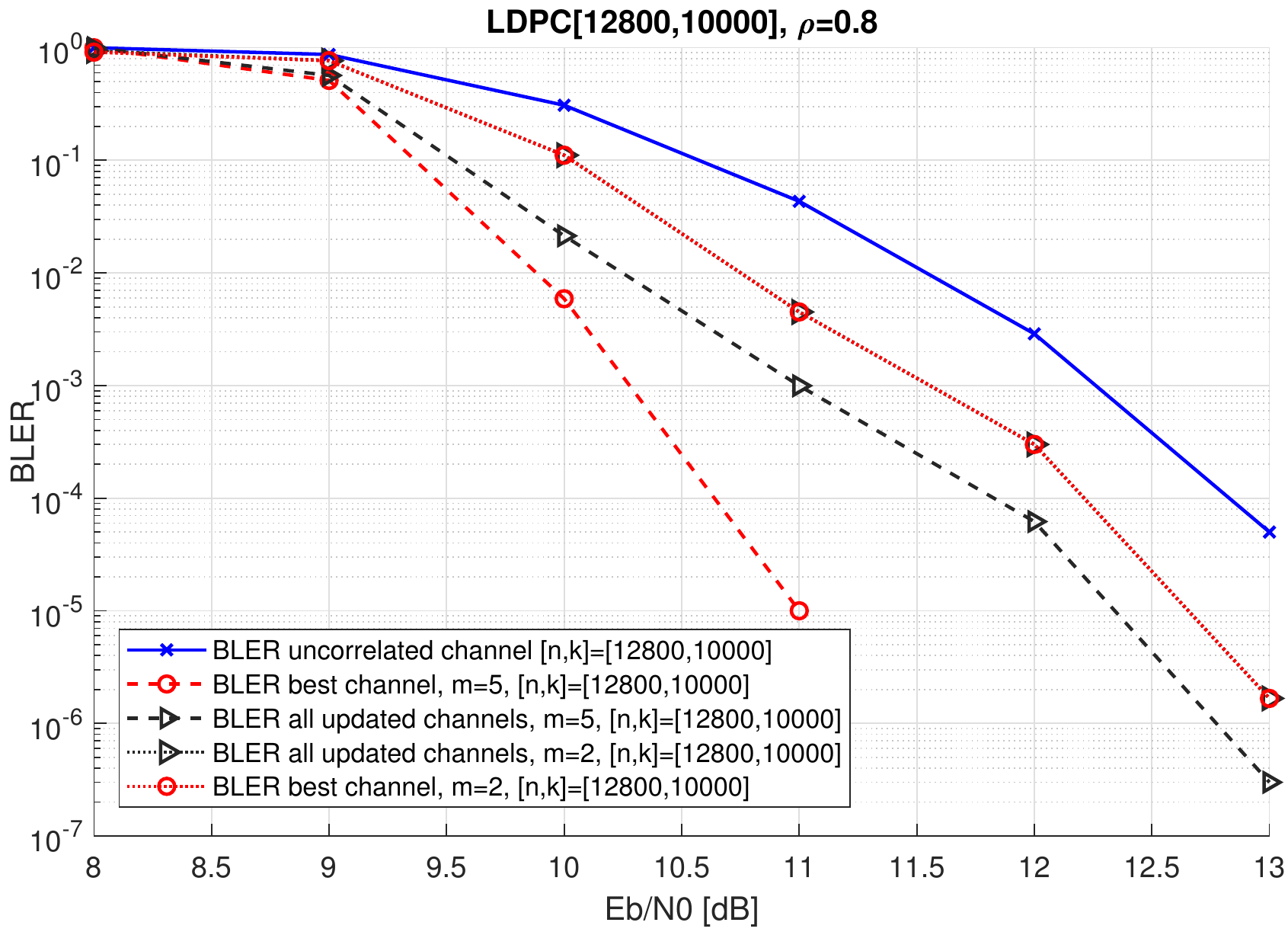}
    \else
    \includegraphics[width=0.6 \columnwidth]{ldpc_m5_2_rate07813_rho08_3}
    \fi
    \caption{BLER vs Eb/N0 for $\sbrace{12800,10000}$ LDPC codes decoded with Belief Propagation and dynamic Noise Recycling in $m=5$ and $m=2$ channels.
	}
    \label{fig:bler_same_rate_m_mg3}
\end{figure}

We next considered decoders in which the proxy for the most confident
decoding could be the channel whose noise estimation is the most
likely. Initially the $i$-th decoder decodes its orthogonal channel
output $\vy_i$ without Noise Recycling, and proceeds to compute the
log-likelihood of the estimated noise $\hat{\vz}_i$,
\begin{equation*}
    l_i(\hat{\vz}_i;0,\sigma) = -\sum_{j=1}^{n} \log\left( \frac{1}{\sigma \sqrt{2\pi}} e^{-\hat{z}_{i,j}^{2}/(2\sigma^2)} \right),
\end{equation*}
for each $i\in \{1, \ldots, m\}$. Using the log-likelihood of the noise estimated at each decoder,
the leading decoder is chosen as the decoder with the most likely
noise channel estimated, namely,
\begin{equation*}
    \underset{(i)}{\arg\max} \quad l_i(\hat{\vz}_i;0,\sigma).
    %f_i(\hat{\vz}_i;0,\sigma) \geq f_{\Tilde{i}}(\hat{\vz}_{\Tilde{i}};0,\sigma), \text{  } \forall i \neq \Tilde{i}
\end{equation*}
The remaining decoders subtract the correlated portion of the
estimated noise from their received signals, starting from orthogonal
channels $i\pm 1$, that leads to higher SNRs, as depicted in
Fig.~\ref{fig:gm_chain}.

Fig.~\ref{fig:bler_same_rate_m_mg3} provides performance results
for $\sbrace{12800,10000}$ LDPCs decoded with Belief Propagation.
Even at these long block-lengths, where one might expect noise
variability between different channels to average out, a gain of
approximately 1.25 dB is observed between a decoder that does not
use Noise Recycling, and a decoder that does so with $m=5$ at a
target BLER of $10^{-3}$.

\subsection{Noise Recycling With Re-Recycling}\label{sub:rerecycling}
In all of the results shown so far, the lead channel pays a price
by not benefitting from Noise Recycling. Here we demonstrate the
potentially counter-intuitive point that this need not be the case.
Instead, one can re-recycle, feeding back a recycled noise to the
lead channel, redecode it, and get improved BLER. Reconsidering the
heuristic argument presented in equation \eqref{eq:bler}, on
re-recycling, this block error rate on the left hand side plays the
role that $\BLER(\sigma^2)$ did in the discussion of recycling. As
a result, this gives additional probability to the
$\BLER(\sigma^2(1-\rho^2))$ term and so a further contraction to a
less noisy channel, improving the lead channel's decoding.

\begin{figure}
    \centering
    \ifdouble
    \includegraphics[trim= 0cm 0cm 0cm 0.4cm,clip,width=1 \columnwidth]{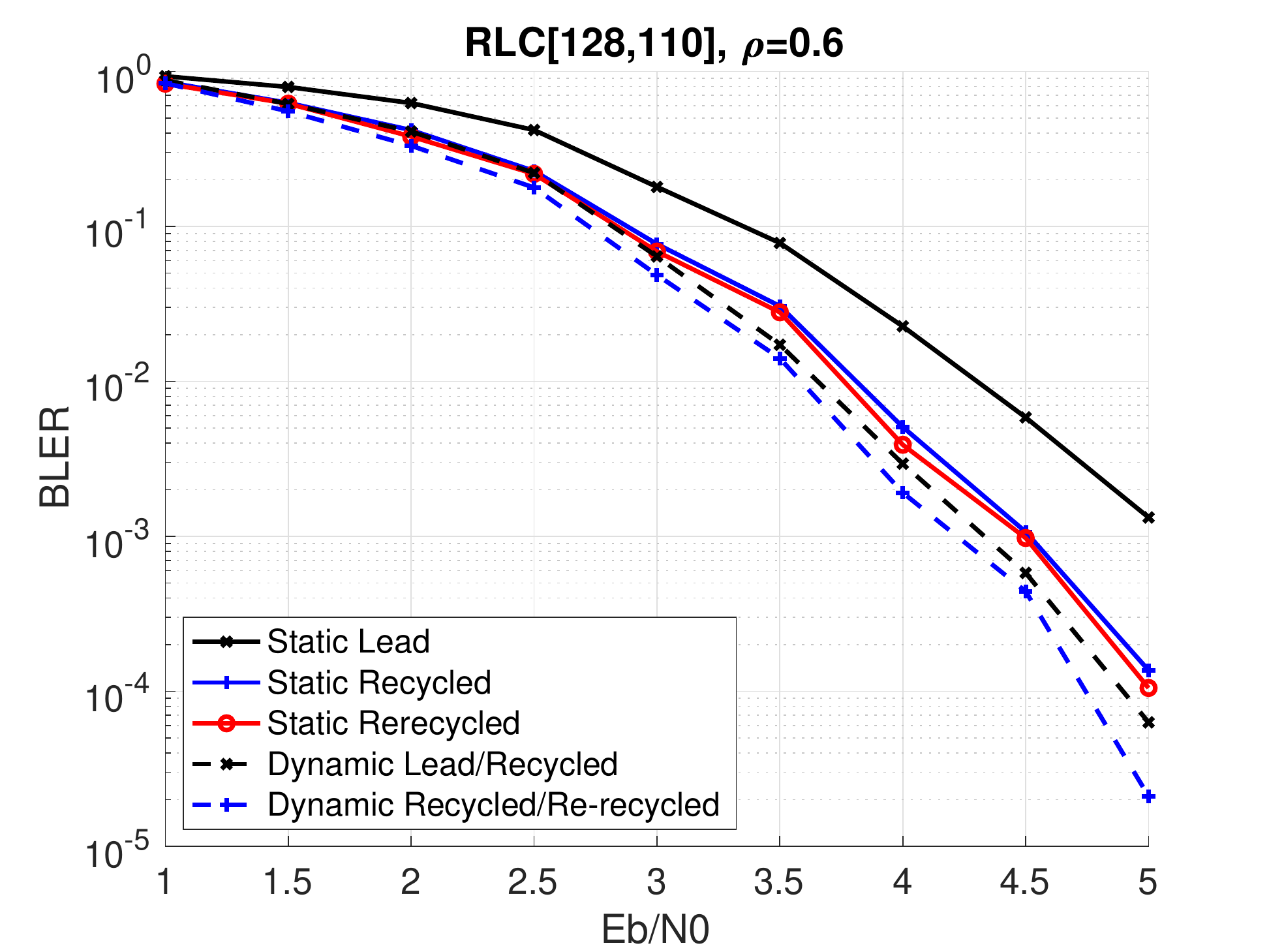}
    \else
    \includegraphics[trim= 0cm 0cm 0cm 0.4cm,clip,width=0.6 \columnwidth]{rerecycling_128_110}
    \fi
    \caption{BLER vs Eb/N0 for two $\sbrace{128,110}$ RLCs decoded with
    ORBGRAND and using Static or Dynamic Noise Recycling with and without
    Re-Recycling for a symmetric Gaussian channel with $\rho=0.6$.}
    \label{fig:rerecycling_128_110}
\end{figure}

For a symmetric Gaussian model with two channels employing RLCs
decoded with ORBGRAND, Fig. \ref{fig:rerecycling_128_110} reports
BLER performance.  For Static Noise Recycling, where a pre-determined
channel is chosen to always lead, the black line reports its BLER
curve, while the blue line reports the BLER curve of the second
channel after Noise Recycling. The noise recycled channel experiences
a $\sim$0.5dB gain at a target BLER of $10^{-3}$. The red curve provides
the BLER performance for the lead channel after Noise Recycling is
used from the second channel and it is re-decoded, and its performance
slightly outstrips the second channel, also giving a $\sim$0.5dB gain.
That is, through Noise Recycling with Re-Recycling, both channels
have gained approximately 0.5dB, and the lead channel is at no
disadvantage. In dashed lines, also shown are the results for the
Dynamic Noise Recycling results with ORBGRAND's query count used
as the proxy for decoding confidence. With Dynamic Noise Recycling,
each channel is the lead $\sim$50\% of the time and the recycled channel
$\sim$50\% of the time, giving the dashed black line.  With Re-Recycling,
the lead channel is replaced with its re-recycled decoding and a
further gain is obtained.

\begin{figure}
    \centering
    \ifdouble
    \includegraphics[width=0.95 \columnwidth]{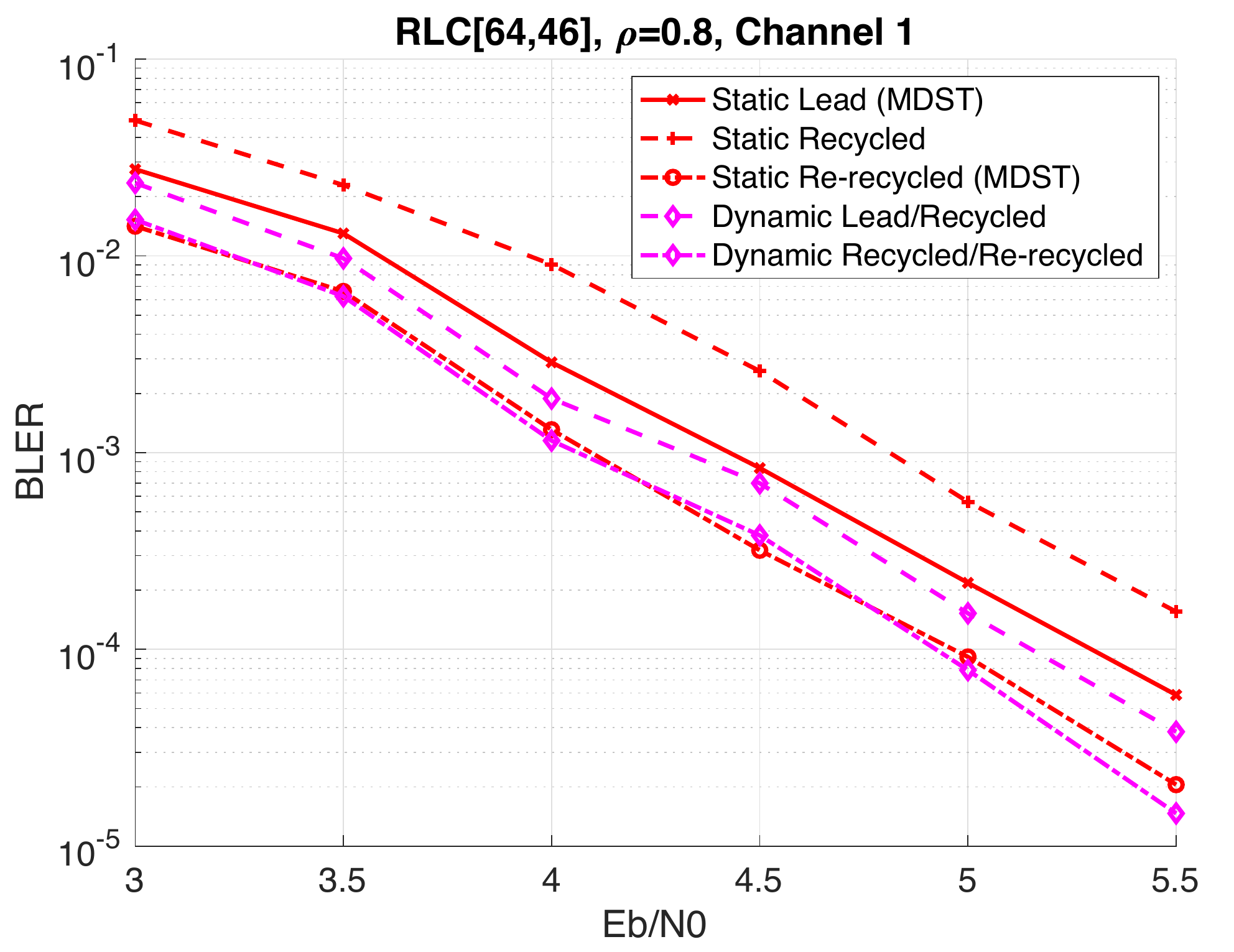}
    \includegraphics[width=0.95 \columnwidth]{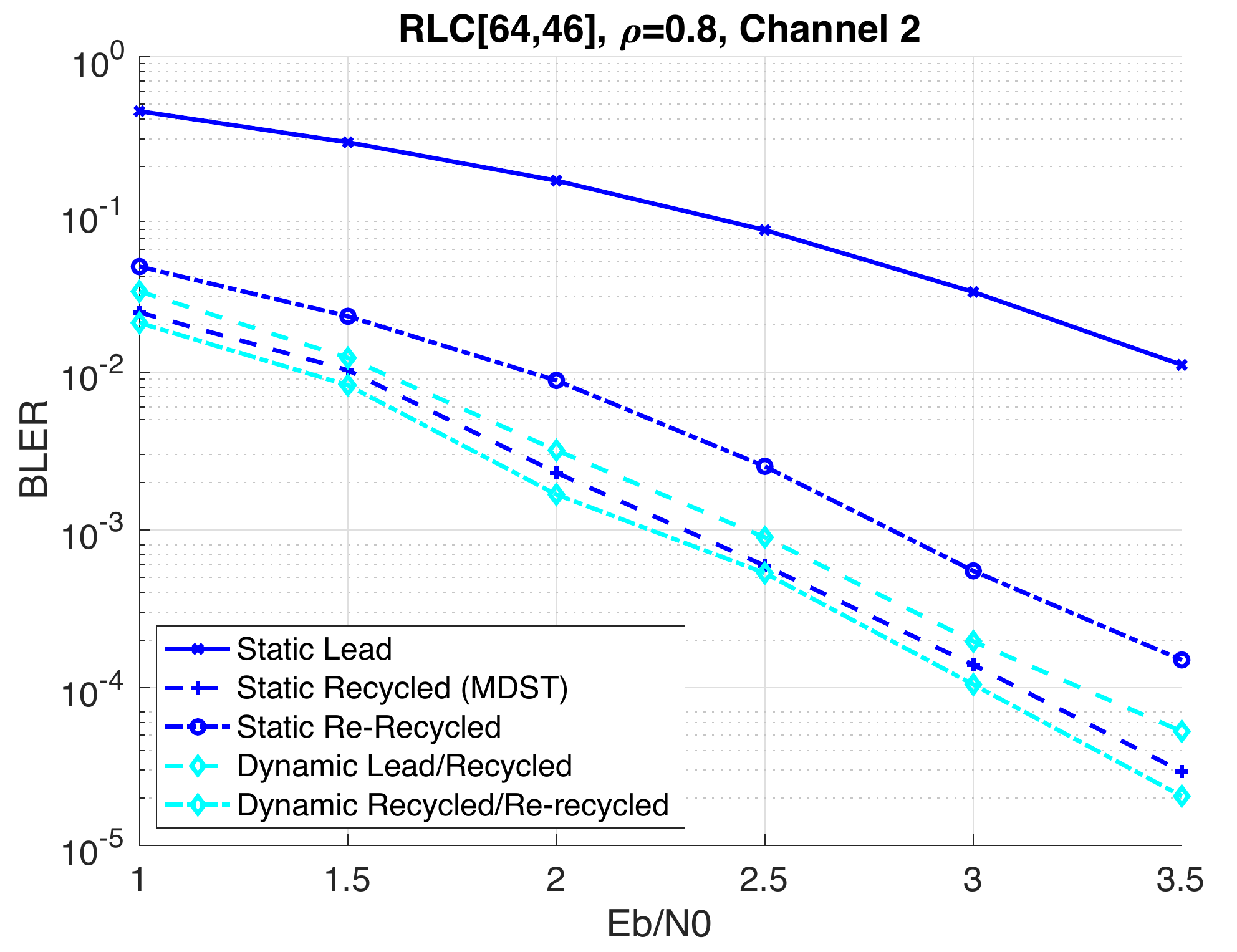}
    \else
    \subfigure{\includegraphics[trim=0cm 0.0cm 0cm 0cm,width=0.49\columnwidth]{rerecycling_64_46_asym_a}}
    \subfigure{\includegraphics[trim=0cm 0.0cm 0cm 0cm,width=0.49\columnwidth]{rerecycling_64_46_asym_b}}
    \fi
    \caption{BLER vs Eb/N0 for two $\sbrace{64,46}$ RLCs decoded with
    ORBGRAND and using Static or Dynamic Noise Recycling with Re-Recycling, where
    channel one is operating at 2dB higher than channel two, and $\rho=0.8$.}
    \label{fig:rerecycling_64_46_asym}
\end{figure}

Fig. \ref{fig:rerecycling_64_46_asym} provides a second example of
re-recycling, but with asymmetric channels. Here two channels use
the same class of $\sbrace{64,46}$ RLCs, but with channel two having
a SNR that is 2dB lower than that experienced on channel one. The
solid red and blue lines show the independent BLER curves of each
channel when decoded independently. If, as advocated in this paper,
MDST was applied to determine a static recycling order Channel 1
would be selected as the lead.  In that case, Channel 1 would have
BLER performance of the solid red line and the performance of Channel
2, which benefits from Noise Recycling, would improve to the dashed
blue line, a gain of over 2dB. If Re-Recycling was used, Channel
1's performance would see the improvement shown in the dash-dotted
red line, giving it a gain of $\sim$0.5dB. That is, by first decoding
the more reliable channel, the less reliable channel's performance
is improved. However, by Re-Recycling noise, the first channel's
performance is then also improved. If MDST was not applied and
Channel 2, the less reliable one, was chosen as the lead, due to
the asymmetry, the performance of Channel 1 degrades by $\sim$0.3dB after
Noise Recycling. This emphasises that in the presence of asymmetries,
static order choice has significant consequences. However, even in
this case, having chosen Channel 2 as the lead in contradiction to
the order determined by MDST, note that its performance is significantly
enhanced by Re-Recycling, swapping the solid blue line for the
dash-dotted blue one.

When Dynamic Noise Recycling is employed with ORBGRAND's query count
used as the soft information that determines decoding confidence,
each channel is a dynamic mixture of being the lead channel and the
recycled one (the dashed magenta and cyan lines) or, if Re-Recycling
is used, the recycled channel and the re-recycled channel (the
dash-dotted magenta and cyan lines). The performance of these is
also shown, where it can be seen that the latter, a dynamically
chosen average of being recycled or re-recycled, gives the best
performance with an approximate 2dB gain for Channel 2 and a $\sim$0.5dB
gain for Channel 1.

\section{Conclusion and Discussion}\label{sec:conclusion}
We introduced Noise Recycling for orthogonal channels experiencing
correlated noise, as a means to improve communication performance
for any combination of codes and decoders. The performance improvement
is twofold, we proved it enables rate gains and provided evidence
of its reliability improvement aspect. We analyzed orthogonal
correlated channels, i.e. channels in which data that is sent on
different channels is independent. A natural extension is considering
the use of Noise Recycling in wireless communications, and the
consequences of uncertainty in it. Noise Recycling points to the
benefit of correlation among orthogonal channels, opening an
interesting vein of investigation where orthogonal channels, say
in OFDM or TDMA, are chosen with a preference for noise correlation
among them, with attendant effects in terms of rate and power
allocation among orthogonal channels. In particular, noise correlation may be seen as an added advantage to dense OFDM channel placement, beyond the inherently desirable efficiency in bandwidth use that density entails.

\appendices

\section{Achievability Proof}\label{sec:code}
\begin{proof}
To establish Theorem \ref{LB}, our coding scheme carries
the flavor of~\cite[Chapter 9.1]{cover2012elements}. We create
$m$ independent random codebooks such that the $j$-th codebook
consists of $2^{nR_j}$ codewords independently drawn from
$\vX^j\paren{1},\ldots,\vX^j\paren{2^{nR_j}}\sim\cN\paren{0,P_j-\epsilon}$, where
$R_j<C\paren{P_j/\paren{1-\rho_j^2}\sigma_j^2}$, and a superscript $j$ indicates that a codeword was
chosen from the $j$-th codebook. For each orthogonal channel, the
codebook need only be known at its encoder and decoder. The transmitters of the
orthogonal channels send $\vX^1\paren{i_1},\ldots,\vX^m\paren{i_m}$.
The decoders operate according to the decoding order dictated by
Lemma~\ref{lem:mdst}. With $\vZhat_{\pj}=0$ if $\pj=r$, the $j$-th
decoder subtracts $\rho_j'\vZhat_{\pj}$ from its channel output
$\vY_j$, resulting in
$\vY_j'=\vY_j-\rho_j'\vZhat_{j}$, where
\begin{equation*}
\rho_j'=\begin{cases}
\rho'_{\pj,j} & \pj\neq r \\
0 & \pj=r
\end{cases}
\end{equation*}
and $\vZhat_{\pj}$ is the estimated noise of orthogonal channel
$\pj$. It then identifies as the decoding the unique codeword that
is jointly-typical with $\vY_j'$ and satisfies the power constraints.
If a codeword does not exist or is not unique, the $j$-th decoder
decodes in error.

We prove the result using techniques redolent of those in~\cite[Chapter
9.1]{cover2012elements}. We bound from below the probability that
the jointly decoding $\vX_1,\ldots,\vX_m$ is successful. The event
of successfully decoding the $j$-th channel and all of its predecessors
in the MDST is denoted by $\cC_j$. The event of successfully decoding
all orthogonal channels is denoted by $\cC$. The event of a decoding
failure in the $j$-th decoder is denoted by $\cE_j$. Define the
following events:
\begin{eqnarray*}
   &&E_{0,j}=\cbrace{n^{-1}\sum_{i}{\paren{X^j\paren{i}}^2}>P_j},\\
   &&E_{i,j}=\cbrace{X_j\paren{i},Y_j'\text{ are jointly
}\epsilon\text{-typical}}.
\end{eqnarray*}
Then $P\paren{\cC}=\prod_{j=1}^m P\paren{\cC_j\mid \cC_{\pj}}$. From results for a single channel, we know that $P\paren{\cC_j\mid\cC_{\pj}}\geq 1-3\epsilon$ when $\pj=r$ for $n$ sufficiently large. Without loss of generality, assume that the $j$-th transmitter sends the first codeword of the $j$-th codebook. We bound $P\paren{\cE_j\mid \cC_{\pj}},\pj\neq r$ in a similar fashion to the single case:
\begin{equation*}
\begin{split}
    &P\paren{\cE_j\mid \cC_{\pj}}\leq \\
    &P\paren{E_{0,j}\mid \cC_{\pj}}+P\paren{E_{1,j}^c\mid \cC_{\pj}}+\sum_{i=2}^{2^{nR_j}}P\paren{E_{i,j}\mid \cC_{\pj}},
\end{split}
\end{equation*}
where the inequality follows from the union bound. For sufficiently large $n$, $P\paren{E_{0,j}\mid \cC_{\pj}}\leq \epsilon$ by the law of large numbers and $P\paren{E_{1,j}^c\mid \cC_{\pj}}\leq \epsilon$ by joint typicality. We bound $P\paren{E_{i,j}\mid \cC_{\pj}},i>1$:
\begin{equation*}
\begin{split}
&P\paren{E_{i,j}\mid \cC_{\pj}}= \\
&P\paren{E_{i,j}\mid X_{\pj},Y_{\pj},X_{\parent{\pj}},Y_{\parent{\pj}},\ldots}\leq\\
&P\paren{E_{i,j}\mid X_{\pj},Y_{\pj}}\leq 2^{-n\paren{\mutualinf{X_j}{Y_j\mid X_{\pj}, Y_{\pj}}-R_j-3\epsilon}}
\end{split}
\end{equation*}and
\begin{equation}\label{rate_j}
\begin{aligned}
&\mutualinf{X_j}{Y_j\mid X_{\pj}, Y_{\pj}}\\
&=\mutualinf{X_j}{Y_j\mid X_{\pj}, X_{\pj}+Z_{\pj}}=\mutualinf{X_j}{Y_j\mid Z_{\pj}}\\
&=h\paren{X_j\mid Z_{\pj}}+h\paren{Y_j\mid Z_{\pj}}-h\paren{X_j,Y_j\mid Z_{\pj}}\\
&=h\paren{X_j}+h\paren{X_j+\rho'_j Z_{\pj}+\Xi_{\pj,j}\mid Z_{\pj}}-\\
&\hspace{1cm}h\paren{X_j,X_j+\rho'_j Z_{\pj}+\Xi_{\pj,j}\mid Z_{\pj}}\\
&=h\paren{X_j}+h\paren{X_j+\Xi_{\pj,j}}-h\paren{X_j,X_j+\Xi_{\pj,j}}\\
&=\mutualinf{X_j}{X_j+\Xi_{\pj,j}}=\mutualinf{X_j}{Y_j'}
\end{aligned}
\end{equation}
using the fact that $\paren{X_j,\Xi_{\pj,j}}\perp Z_{\pj}$. Therefore,
\begin{equation*}
   P\paren{E_{i,j}\mid X_{\pj}, Y_{\pj}}\leq 2^{-n\paren{\mutualinf{X_j}{Y_j'}-R_j-3\epsilon}}.
\end{equation*}
Picking $R_j<\mutualinf{X_j}{Y_j'}-3\epsilon$ yields $P\paren{\cE_{i,j}\mid \cC_{\pj}}\leq 3\epsilon$. Ultimately, we get $P\paren{\cC}\geq \paren{1-3\epsilon}^m$ which concludes the proof as $\epsilon$ can be made arbitrarily small.
\end{proof}

\section{Upper Bound Proof}\label{sec:UpperBound_proof}
In this section, we derive the proof of Theorem~\ref{eq:upper_bound}:
\ifdouble
\begin{eqnarray}
    &&\hspace{-0.5cm} I(X_{j},X_{i};Y_{j},Y_{i}) \nonumber\\
    &&\hspace{-0.5cm} = I(X_{j},X_{i};Y_{i}) + I(X_{j},X_{i};Y_{j}|Y_{i}) \nonumber\\
    &&\hspace{-0.5cm} = I(X_{i};Y_{i}) + I(X_{j};Y_{i}|X_{i}) + I(X_{j},X_{i};Y_{j}|Y_{i}) \nonumber\\
    &&\hspace{-0.5cm} \stackrel{(a)}{=} I(X_{i};Y_{i}) + I(X_{j},X_{i};Y_{j}|Y_{i}) \nonumber\\
    &&\hspace{-0.5cm} = I(X_{i};Y_{i}) + I(X_{i};Y_{j}|Y_{i}) + I(X_{j};Y_{j}|X_{i},Y_{i}) \nonumber\\
    &&\hspace{-0.5cm} \stackrel{(b)}{=} I(X_{i};Y_{i}) + I(X_{i};Y_{j}|Y_{i}) + I(X_{j};Y_{j}') \nonumber\\
    &&\hspace{-0.5cm} = I(X_{i};Y_{i}) + h(Y_{j}|Y_{i}) - h(Y_{j}|X_{i},Y_{i}) + I(X_{j};Y_{j}') \nonumber\\
    &&\hspace{-0.5cm} = I(X_{i};Y_{i}) + h(Y_{j}|Y_{i})  - h(X_{j}+\rho'_{i,j} Z_{i}+\Xi_j|Z_{i}) \nonumber\\
    &&\hspace{6.0cm} + I(X_{j};Y_{j}') \nonumber\\
    &&\hspace{-0.5cm} = I(X_{i};Y_{i}) + h(Y_{j}|Y_{i}) - h(Y_{j}') + I(X_{j};Y_{j}') \nonumber\\
    &&\hspace{-0.5cm} = I(X_{i};Y_{i}) + h(Y_{j}) - I(Y_{j};Y_{i}) - h(Y_{j}') + I(X_{j};Y_{j}') \nonumber\\
    && \hspace{-0.5cm} \stackrel{(c)}{=} I(X_{i};Y_{i}) + h(Y_{j}) \nonumber\\
    && \hspace{2.1cm} + \frac{1}{2}\log\left(1-\tilde{\rho}_{i,j}^2\right) - h(Y_{j}') + I(X_{j};Y_{j}') \nonumber\\
    && \hspace{-0.5cm} \stackrel{(d)}{=} \frac{1}{2}\log\left(1+\frac{P_i}{\sigma_{i}^{2}}\right) + \frac{1}{2}\log\left(2\pi e\left(P_j+\sigma_{j}^{2}\right)\right) \nonumber\\
    && \hspace{-0.1cm} + \frac{1}{2}\log\left(1-\tilde{\rho}_{i,j}^2\right) - \frac{1}{2}\log\left(2\pi e\left(P_j+\paren{1-\rho_{i,j}^2}\sigma_{j}^{2}\right)\right)\nonumber\\
    && \hspace{3.7cm} + \frac{1}{2}\log \left(1+\frac{P_j}{\paren{1-\rho_{i,j}^{2}}\sigma_{j}^2}\right) \nonumber\\
    && \hspace{-0.5cm} = \frac{1}{2}\log\left(1+\frac{P_i}{\sigma_{i}^{2}}\right) + \frac{1}{2}\log\left(\frac{P_j+\sigma_{j}^{2}}{P_j+\paren{1-\rho_{i,j}^2}\sigma_{j}^2}\right)\nonumber\\
    && \hspace{1.1cm} + \frac{1}{2}\log\left(1-\tilde{\rho}_{i,j}^2\right) + \frac{1}{2}\log \left(1+\frac{P_j}{\paren{1-\rho_{i,j}^{2}}\sigma_{j}^2}\right) \nonumber
\end{eqnarray}
\else
\begin{eqnarray}
    &&\hspace{-0.5cm} I(X_{j},X_{i};Y_{j},Y_{i}) \nonumber\\
    &&\hspace{-0.5cm} = I(X_{j},X_{i};Y_{i}) + I(X_{j},X_{i};Y_{j}|Y_{i}) \nonumber\\
    &&\hspace{-0.5cm} = I(X_{i};Y_{i}) + I(X_{j};Y_{i}|X_{i}) + I(X_{j},X_{i};Y_{j}|Y_{i}) \nonumber\\
    &&\hspace{-0.5cm} \stackrel{(a)}{=} I(X_{i};Y_{i}) + I(X_{j},X_{i};Y_{j}|Y_{i}) \nonumber\\
    &&\hspace{-0.5cm} = I(X_{i};Y_{i}) + I(X_{i};Y_{j}|Y_{i}) + I(X_{j};Y_{j}|X_{i},Y_{i}) \nonumber\\
    &&\hspace{-0.5cm} \stackrel{(b)}{=} I(X_{i};Y_{i}) + I(X_{i};Y_{j}|Y_{i}) + I(X_{j};Y_{j}') \nonumber\\
    &&\hspace{-0.5cm} = I(X_{i};Y_{i}) + h(Y_{j}|Y_{i}) - h(Y_{j}|X_{i},Y_{i}) + I(X_{j};Y_{j}') \nonumber\\
    &&\hspace{-0.5cm} = I(X_{i};Y_{i}) + h(Y_{j}|Y_{i})  - h(X_{j}+\rho'_{i,j} Z_{i}+\Xi_j|Z_{i}) + I(X_{j};Y_{j}') \nonumber\\
    &&\hspace{-0.5cm} = I(X_{i};Y_{i}) + h(Y_{j}|Y_{i}) - h(Y_{j}') + I(X_{j};Y_{j}') \nonumber\\
    &&\hspace{-0.5cm} = I(X_{i};Y_{i}) + h(Y_{j}) - I(Y_{j};Y_{i}) - h(Y_{j}') + I(X_{j};Y_{j}') \nonumber\\
    && \hspace{-0.5cm} \stackrel{(c)}{=} I(X_{i};Y_{i}) + h(Y_{j}) + \frac{1}{2}\log\left(1-\tilde{\rho}_{i,j}^2\right) - h(Y_{j}') + I(X_{j};Y_{j}') \nonumber\\
    && \hspace{-0.5cm} \stackrel{(d)}{=} \frac{1}{2}\log\left(1+\frac{P_i}{\sigma_{i}^{2}}\right) + \frac{1}{2}\log\left(2\pi e\left(P_j+\sigma_{j}^{2}\right)\right) \nonumber\\
    && \hspace{2.1cm} + \frac{1}{2}\log\left(1-\tilde{\rho}_{i,j}^2\right) - \frac{1}{2}\log\left(2\pi e\left(P_j+\paren{1-\rho_{i,j}^2}\sigma_{j}^{2}\right)\right) + \frac{1}{2}\log \left(1+\frac{P_j}{\paren{1-\rho_{i,j}^{2}}\sigma_{j}^2}\right) \nonumber\\
    && \hspace{-0.5cm} = \frac{1}{2}\log\left(1+\frac{P_i}{\sigma_{i}^{2}}\right) + \frac{1}{2}\log\left(\frac{P_j+\sigma_{j}^{2}}{P_j+\paren{1-\rho_{i,j}^2}\sigma_{j}^2}\right) + \frac{1}{2}\log\left(1-\tilde{\rho}_{i,j}^2\right) + \frac{1}{2}\log \left(1+\frac{P_j}{\paren{1-\rho_{i,j}^{2}}\sigma_{j}^2}\right) \nonumber
\end{eqnarray}
\fi
where (a) follows since $I(X_{j};Y_{i}|X_{i})=0$ using the fact that $\vX_j\perp\vZ_{i}$, and the fact that  $X_{j}$ and $X_{i}$ are independent, (b) follows from \eqref{rate_j}, (c) follows from \cite[Example 8.5.1]{cover2012elements} where $X \sim\cN\paren{0,P}$ and $Y\sim\cN\paren{0,P+\sigma^2}$, such that
\begin{equation*}
    \text{cov}\paren{Y_{i},Y_{j}}=E\paren{Y_{i}Y_j}=E\paren{Z_{i}Z_j}=\rho_{i,j}\sigma_i\sigma_j,
\end{equation*}
and the mutual information between correlated Gaussian random variables with correlation
\begin{align*}
\tilde{\rho}_{i,j}=\rho_{i,j}\frac{\sigma_i\sigma_j}{\sqrt{(P_i+\sigma_{i}^2)(P_j+\sigma_{j}^2)}}
\end{align*}
is
\begin{equation*}
 I(Y_{j};Y_{i}) = h(Y_{j})+h(Y_{i}) - h(Y_{j},Y_{i}) = -\frac{1}{2}\log \left(1-\tilde{\rho}_{i,j}^2\right),
\end{equation*}
(d) follows from \cite[Theorem 8.6.5]{cover2012elements}, where using the fact that $X_j$ and $Z_j$ are independent,
\begin{equation*}
    h(Y_j) \leq \frac{1}{2} \log\left(2\pi e\left(P_j+\sigma_{j}^{2}\right)\right),
\end{equation*}
with equality if and only if $X_j \sim\cN\paren{0,P_j}$. This completes the upper bound proof.
\qed

\bibliographystyle{IEEEtran}
\bibliography{bib}
\end{document}